\pdfoutput=1
\documentclass[
    aps,
    pra,
    longbibliography,
    superscriptaddress,
    amsmath,amssymb,amsfonts, 
    tightenlines,
    twocolumn
    ]{revtex4-1}
\usepackage{ifthen}

\newif\ifarxiv
\newif\ifjournal
\journalfalse \arxivtrue  % comment this in for the arxiv version
\usepackage{bibunits}
\defaultbibliography{bib_analogue_rb}
\defaultbibliographystyle{./myapsrev4-2}
\usepackage{floatrow}
\newfloatcommand{capbtabbox}{table}[][\Xhsize]

\usepackage[T1]{fontenc} %256 bit font encoding
\usepackage[english]{babel} %english language
\usepackage[utf8]{inputenc}
\usepackage{amsthm}
\usepackage{times}
\usepackage{amssymb}

\usepackage[
    pdfencoding=auto,
    psdextra,
    colorlinks=true,
    hypertexnames=false
            ]{hyperref}
\PassOptionsToPackage{linktocpage}{hyperref} %link numbers in TOC for arxiv

\hypersetup{
  colorlinks   = true, %Colour links instead of ugly boxes
  urlcolor = teal,
  linkcolor    = blue, %Colour of internal links
  citecolor=green!60!black %Colour of citations
}
\usepackage{bookmark} %to fix the pdfbookmarks
\usepackage{booktabs}

\usepackage{graphicx} 
\usepackage{subfigure}

\usepackage{mathtools,bbm}
\usepackage{dsfont}
\usepackage{comment}
\usepackage{enumerate}
\usepackage{tikz}
\usepackage{tabularx}
\usepackage[export]{adjustbox}
\usepackage{ytableau}

\usepackage{algorithm}
\usepackage{algpseudocode}
\makeatletter
\renewcommand{\ALG@name}{Protocol}

\newtheorem{corollary}{Corollary}
\makeatother
\usepackage{xcolor}
\usepackage{color, colortbl}

\newcounter{protocol}

\newcommand\appendixtableofcontents{%
  \begingroup
  \let\clearpage\relax
  \tableofcontents
  \endgroup
}

\newtheorem{theorem}{Theorem}%[section]
\newtheorem{lemma}[theorem]{Lemma}

\DeclareMathOperator{\tr}{Tr}
\newcommand{\ket}[1]{\vert{#1}\rangle}
\newcommand{\bra}[1]{\langle{#1}\vert}

\newcommand{\norm}[1]{\Vert #1 \Vert}

\newcommand{\ee}{\mathrm{e}}

\DeclareMathOperator{\Sym}{Sym}

\makeatletter % metadata for hyperref
\hypersetup{pdftitle = {Benchmarking bosonic and fermionic dynamics},
	     pdfauthor = {
      Jadwiga Wilkens,
      Marios Ioannou,
      Ellen Derbyshire,
      Jens Eisert,
      Dominik Hangleiter,
      Ingo Roth, Jonas
      Haferkamp
      },
	     pdfsubject = {
      quantum physics,
      randomized analog benchmarking,
      benchmarking protocol
      },
 	     pdfkeywords = {
       quantum,
       randomized benchmarking,
       characterisation,
       analog,
       bosons,
       fermions,
       quantum simulation,
       continuous variable,
       representation theory
       }
	    }
\makeatother

\newcommand{\fu}{Dahlem Center for Complex Quantum Systems, Freie Universit\"{a}t Berlin, Germany}
\newcommand{\quics}{Joint Center for Quantum Information and Computer Science (QuICS), NIST/University of Maryland, College Park, MD 20742, USA}
\newcommand{\hzb}{Helmholtz-Zentrum Berlin f{\"u}r Materialien und Energie, 14109 Berlin, Germany}
\newcommand{\tii}{Quantum Research Center, Technology Innovation Institute, Abu Dhabi}
\newcommand{\jku}{Johannes Kepler University Linz,
Institute for Integrated Circuits, 4040 Linz, Austria}

\usepackage[normalem]{ulem}

\begin{document}

\title{Benchmarking bosonic and fermionic dynamics}

\def\byline{Contributed equally and correspond under wilkensjadwiga@gmail.com and marios\_ioannou@outlook.com.}
\def\lastbyline{Contributed equally.}

\author{J. Wilkens}\thanks{\byline}
\affiliation{\fu}
\affiliation{\tii}
\affiliation{\jku}
\author{M. Ioannou}\thanks{\byline}
\affiliation{\fu}
\author{E. Derbyshire}
\affiliation{\fu}
\author{J. Eisert}
\affiliation{\fu}
\affiliation{\hzb}

\author{D. Hangleiter}
\affiliation{\quics}

\author{I. Roth}\thanks{\lastbyline}
\affiliation{\tii}
\author{J. Haferkamp}\thanks{\lastbyline}
\affiliation{School of Engineering and Applied Sciences, Harvard University, Cambridge, MA02318, USA }

\begin{abstract}
Analog quantum simulation allows for assessing static and dynamical properties of strongly correlated quantum systems to high precision. 
To perform 
simulations outside the reach of classical computers, accurate and reliable implementations of the anticipated Hamiltonians are required. 
To achieve those,
characterization and benchmarking tools are a necessity. 
For digital quantum devices, randomized benchmarking can provide a benchmark on the average quality of the implementation of a gate set. In this work, we introduce a versatile framework for \emph{randomized analog benchmarking} of bosonic and fermionic quantum devices implementing particle number preserving dynamics. 
The scheme makes use of the restricted operations which are native to analog 
simulators and other continuous variable systems.
Importantly, like randomized benchmarking, it is robust against state preparation and measurement errors. 
We discuss the scheme's efficiency,  derive theoretical performance guarantees and showcase the protocol with numerical examples.
\end{abstract}

\maketitle
\ifjournal
\begin{bibunit}
\fi

Analog quantum simulators are physical systems designed to accurately implement an idealized Hamiltonian model under precisely controlled conditions and with tunable parameters. 
Particularly prominent examples of quantum simulations investigate bosonic Hamiltonians such as the Bose-Hubbard model which describes cold-atoms in optical lattices~\cite{jaksch_cold_1998,
greiner_quantum_2002,
Trotzky,
koehl_fermionic_2005,
choi_exploring_2016,
GrossHubbard}
and certain superconducting-qubit systems~\cite{roushan_spectroscopic_2017,
0706.0212,
HartreeFock,
IBM2023}. 
In contrast to universal, fault-tolerant quantum computers performing arbitrary quantum computations, analog quantum simulators are special-purpose devices \cite{hangleiter_analogue_2022}. 
By probing the physics of certain specific models, they might shed light on fundamental questions, such as the mechanism behind high-temperature superconductivity~\cite{koehl_fermionic_2005}, why there is unusual magnetism in Hubbard chains~\cite{GrossHubbard}, and the persistence of many-body localization~\cite{choi_exploring_2016}.

To deliver on this promise, the simulated model needs to accurately describe the analog simulator~\cite{hangleiter_analogue_2022}.
Traditionally, this has been ensured for specific experiments following the ``predict-and-compare paradigm'' wherein an experiment is classically simulated and the results compared with the data. 
However, assessing the quality of an analog simulator as a computational tool requires comparable performance metrics that are independent of a specific experiment or device.
For universal digital simulators, a variety of platform-independent benchmarks for individual qubits and gates 
have been developed~\cite{BenchmarkingReview, KlieschRoth:2020:Tutorial}, 
most prominently, \emph{randomized benchmarking (RB)} of entire gate-sets~\cite{EmeAliZyc05, KnillBenchmarking,LevLopEme07, DanCleEme09,HelsenEtAl:2020:GeneralFramework}. 
In addition, techniques to check that the individual quantum gates and measurements 
can be faithfully composed in space and time have been developed. This includes
tests for Markovianity~\cite{Markov,Markov,breuer_measure_2009,SusanaNoMarkov,wallman_randomized_2014} and absence of crosstalk~\cite{rudinger_probing_2019,sarovar_detecting_2020}. 
These benchmarking tools are largely lacking for the analog dynamics of bosonic and fermionic systems.
In contrast to universal quantum computers, the goal of an analog simulation is to simulate a physical model under certain restrictions and in a certain parameter regime. Therefore, any RB method must take into account those restrictions. 
First proposals toward this goal have considered the adaptation of digital RB ~\cite{Derbyshire2019RBanalog, shaffer_benchmarking_2020} and RB-like \cite{valahu2024benchmarking} protocols to analog systems, and also developed bespoke RB protocols for fermionic matchgates~\cite{MatchgateBenchmarking}. 

\begin{figure*}
\centering
\includegraphics[width=\linewidth]{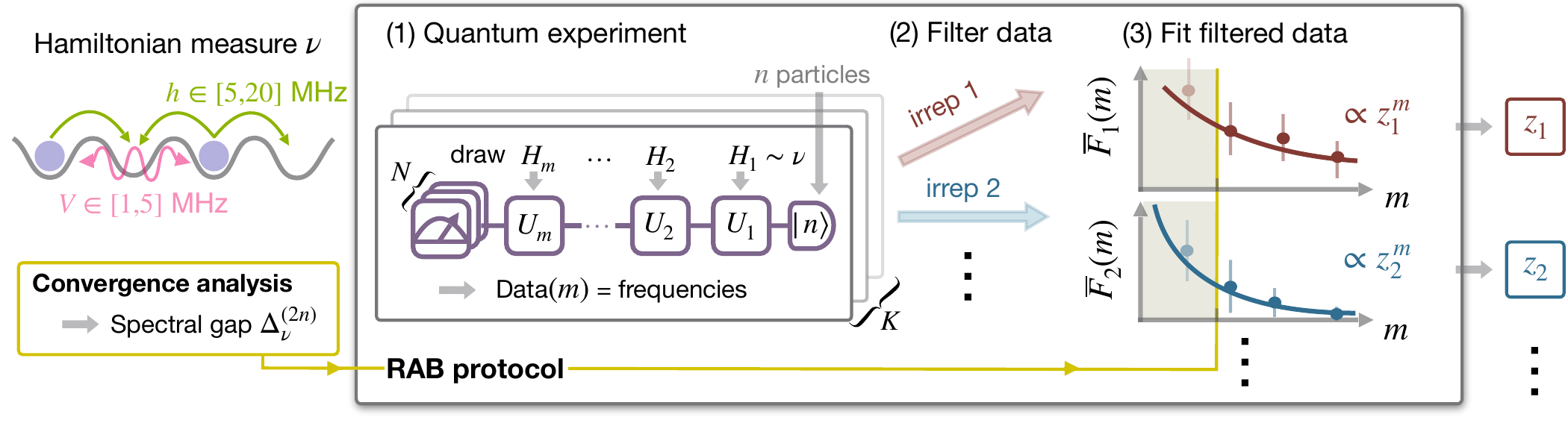}
    \caption{
    \emph{Schematic for Randomized analog benchmarking (RAB).} 
    (a) The RAB protocol can be applied to an experimental system described by a Hamiltonian, say the Bose-Hubbard model \eqref{eq:quartic hamiltonian}, with tunable parameters. 
    We fix a measure $\nu$ according to which the Hamiltonian parameters are chosen, say from some experimentally accessible range.
    Here, we give the example of a single hopping and interaction parameter $h$ and $V$, respectively.
    Given the measure $\nu$, in pre-processing, we perform a convergence analysis that will determine the experimental parameters and the post-processing of the data. 
    (b) The RAB protocol itself then consists of three steps. The first step (1) is the quantum experiment. 
    In the experiment, instances $H_1, \ldots, H_m$ are drawn according to $\nu$ for values of $m \in \mathbb N $ determined by the pre-processing. The experiment is initialized in the $n$-particle sector and evolves the state for a fixed time $\delta t$ under the Hamiltonians $H_i$ giving rise to unitaries $U_i = \exp(- i H \Delta t)$. 
    We then measure in a fixed basis and collect the frequencies of outcomes. 
    In the second step (2), this data is then filtered according to the irreps of the Hamiltonian measure, giving rise to $m$-dependent quantities $\overline F_1, \overline F_2, \ldots$. 
    Finally, (3) each of these quantities is then fitted as $\overline F_i(m) \propto z_i^m$. 
    The parameters $z_i$ are the output of the protocol and quantify the implementation quality.
    }
    \label{fig:protocol simulation} 
\end{figure*}

In this work, we develop and analyze a benchmarking protocol for bosonic and fermionic systems evolving under arbitrary particle-number preserving dynamics.
Our \emph{randomized analog benchmarking (RAB)} protocol fits into the general framework of random sequence protocols~\cite{GateSets, HelsenEtAl:2020:GeneralFramework, heinrich2023randomized} and gives rise to precisely reproducible figures of merit. 
In terms of control, the RAB protocol requires the ability (1) to prepare some input state with a fixed particle number, (2) the ability to perform sequences of random Hamiltonian quenches -- the sequential time evolution for different
parameters of the simulator Hamiltonian -- and (3) a projective measurement, e.g. in the Fock basis.
The remainder of the protocol consists only of classical post-processing. 

For the protocol to give rise to bona-fide benchmarks, it is necessary that the random time evolutions thoroughly explore the relevant subspaces. 
Technically, this condition is captured by being an approximate group design for the irreducible representations of the group generated by the benchmarked operations. 
This ensures that on each subspace the noise of the dynamics is described by a single parameter \cite{heinrich2023randomized}, a generalization of the effective depolarizing strength \cite{KlieschRoth:2020:Tutorial}. 
We reduce the condition 
to the more common notion of approximate unitary designs~\cite{gross_evenly_2007}, here restricted to the subspace of $n$ particles. 
We show numerically that random quenches of Bose-Hubbard Hamiltonians with interaction and hopping strengths drawn from experimentally accessible distributions, fulfil this condition for sufficiently long sequences; extending previous results along these lines~\cite {vermersch_unitary_2018}. 
We also characterize the representation theory for the group of interacting and non-interacting particle-number-preserving dynamics, which is necessary to implement a protocol of this kind, see Refs.~\cite{HelsenNewEfficientRB, MatchgateBenchmarking, heinrich2023randomized}. 
These latter steps significantly reduce an experimentalist's work to benchmark their device. 
We also provide open-access code~\cite{Jadwiga_github} that can be used to check the design property and implement our protocol for different families of Hamiltonians.

Our protocol's main advantage over existing analog RB protocols~\cite{Derbyshire2019RBanalog, shaffer_benchmarking_2020, valahu2024benchmarking} is its simplicity and its applicability to both bosonic and fermionic systems. 
In particular, the requirements for the physical implementation are mild and compatible with the capabilities of state-of-the-art analog devices~\cite{Taballione_2021}. Importantly, in contrast to other RB schemes 
we \emph{do not} require the implementation of an inversion gate at the end of a random sequence. This is especially relevant for analog quantum simulators where inversion is akin to reversing the sign of field and interaction terms.
The latter is either difficult or adds an overhead. Furthermore, we only impose the general assumptions of Markovianity and time-independence on the noise channel and require no knowledge of the noise model, unlike Ref.~\cite{valahu2024benchmarking}. We note that an analogous benchmarking protocol using uniformly random Gaussian operations in fermionic simulators has been developed in Ref.~\cite{MatchgateBenchmarking}, but this protocol does not generalize to bosonic systems and the particle-number conserving setting we study here. 

It is also worthwhile to compare our protocol to other means of characterizing analog simulators \cite{elben_randomized_2023}. 
Most importantly, one can also perform \emph{Hamiltonian learning} in order to characterize manifestations of analog dynamics, and thereby certify correctness. Notable recent experiments include Hamiltonian learning concerning bosonic evolution~\cite{Hangleiter_robustly_2024} and the characterization of so-called entanglement Hamiltonians~\cite{kokail_entanglement_2021}. Another family of benchmarks applicable to states are those based on randomized measurements~\cite{elben_randomized_2023, brydges_probing_2019, elben_cross-platform_2020}.
While these protocols also make use of random quenches, their goal is to infer properties of a fixed quantum state rather than benchmark the overall performance of a simulator.

Our protocol is designed for systems with few particles. 
This is similar to the ``low complexity'' digital approach where RB is typically performed for one- or two-qubit gates. 
However, for analog devices, our protocol can be considered more holistic and the figures of merit measure the contribution
of the average noise in Hamiltonian evolution.
To quantify the experimental resources needed in the low-complexity regime we give evidence for the sample efficiency of the protocol in terms of the number of modes in the system. 
We do so both analytically and by performing numerical experiments. 

\paragraph*{Particle-number preserving dynamics.}
The quantum system we study evolves under the particle-preserving Hubbard model with a fixed number of modes within the lattice. 
The excitations are either bosonic or fermionic. 
We consider this restricted model because it aligns with current experiments, where nearest-neighbour hopping terms and on-site interactions are accessible. Within the allowed Fock-space $\mathbb F_{d,n}$ of $n$ particles in $d$ modes, the Hamiltonian reads
\begin{align}
	H(h,V) = \sum_{\langle i,j\rangle}h_{i,j} a_i^\dagger a_j + \sum_{i}^d V_i a_i^\dagger a_i^\dagger a_i a_i + h.c.\,,
    \label{eq:quartic hamiltonian}
\end{align}
with $h\in\mathbb C^{d\times d}$ and $V\in\mathbb C^{d}$ representing the hopping terms driving the excitation transport along edges $\langle i, j\rangle$ of some graph and the interaction strength respectively, and $a$ ($a^\dagger)$ the bosonic or fermionic annihilation (creation) operators.
We will focus entirely on the bosonic setting in the main text and comment on the analogous fermionic setting in App.~\ref{app:fermionic}.

We want to benchmark the ability of the simulator to perform quenches for a time $\Delta t$ with different Hamiltonian 
\begin{align}\label{eq:random-gate}
    U = \ee^{-i\Delta t H( h,  V)}\,. 
\end{align}
Drawing the parameter $h$,$V$ of the Hamiltonian and thereby $H$ at random, induces a probability measure $\nu$ on the unitaries acting on the Fock space. 
Since we are restricted to sample $h$, $V$ from the permissible and interesting range of the experiment, $\nu$ is typically not the Haar measure on the group of particle-preserving unitaries.

As an example, we consider the one-dimensional Bose-Hubbard model with site-dependent uniformly random local potentials and
translationally-invariant random hopping
as well as constant on-site interaction terms, uniformly drawn from intervals.

\paragraph*{Randomized analog benchmarking (RAB) protocol.} 
In the following, we give an overview of the RAB protocol, a variant of the filtered RB protocol~\cite{HelsenEtAl:2020:GeneralFramework,heinrich2023randomized}, which consists of two phases: data acquisition and classical post-processing. The details are shown in Fig.~\ref{fig:protocol simulation}. 

\emph{Data acquisition:} Choose a set of sequence lengths $\mathcal M$, a number $K$ of sequences per length, and a number of measurement shots $N$ per sequence. 
These parameters depend on the required accuracy and degree of experimental control, as discussed later. 
The experimental primitive consists first of applying $m \in \mathcal M$ random quenches $U$ drawn according to $\nu$, to an initial state $\rho_0 = |\psi_0 \rangle\langle\psi_0|$.
One then collects statistics of projective measurements in the Fock basis $\{E_x = |x \rangle\langle x| \}_{x=1}^{d_F}$ with $N$ shots and records the relative outcome frequencies $\hat p$. 
The sequence length $m$, the overall sequence of quenches $U = U_m\cdots U_1$ and $\hat p$ are recorded for post-processing. The primitive is repeated for $K$ random sequences of length $m$ for all $m \in \mathcal{M}$.

\emph{Post-processing:} For each sequence of quenches we compute the empirical estimate of the \emph{multishot filter function}  
\begin{align}
    \hat F_\lambda(m) \coloneqq n_\lambda^{-1} \sum_{x=1}^{d_F} f_\lambda(x,U) \hat p_{x},
    \label{eq:Fhat-lambda}
\end{align}
where $\hat p$ is the recorded measurement frequency for the evolution $U$.
The weighting factor $f_\lambda(x,U)$, the \emph{filter function} \cite{heinrich2023randomized}, is defined for the irreducible representations (irrep) $\lambda$ of the group generated by the used ensemble of Hamiltonians. 
The general expression for the filter function reads
\begin{align}
    f_\lambda(x,U) \coloneqq s_\lambda^{-1}\tr(P_\lambda(E_x) U\rho_0U^\dagger)\,, 
    \label{equ:single-filter-function-expression-main-text}
\end{align}
with $P_\lambda$ the projector onto $\lambda$. 
The normalization constants $n_\lambda$ and $s_\lambda$ are determined by the overlap with $\lambda$ of the initial state and the measurement, respectively.
An overview of these constants is given in Table~\ref{table:summary-protocol-ingredients}.

We then compute the average $\overline{F}_\lambda (m)$ of $\hat F_\lambda(m)$ over the $K$ random sequences with identical $m$ and fit the signal to a scalar exponential decay 
\begin{equation}\label{eq: exponential_decay}
    \overline{F}_\lambda (m) \approx_{\text{fit}} A_\lambda z_\lambda^m .
\end{equation}
In the case where the Hamiltonians generate unitaries drawn according to the Haar measure, we can fit directly to the above. In the case where the Hamiltonians are drawn according to some other, less uniform, measure $\nu$, the fitting procedure can only be performed after a threshold sequence length $m_{\mathrm{ths}}$, as explained in the next section. 
The protocol is summarized in~Fig.~\ref{fig:protocol simulation}. 

\begin{table}
\centering
\renewcommand{\arraystretch}{1.5} 
\begin{tabular}{l | c c}
\toprule
& interacting & non-interacting \\
\midrule
$\lambda$ &  $n \in \mathbb{N}^+$ &  $(n,l)$ with $n \in \mathbb{N}^+$, $l\in[n]$ \\
$n_\lambda$ & $(1 - 1/d_F)$ & 
{$\!\begin{aligned} \sum_{a_\lambda,x}\left(C_{x,\tilde x}^{a_\lambda}\right)^2
|\langle x, \psi_0\rangle|^2 \end{aligned}$}
\\
\addlinespace 
$s_\lambda$ &
$(1 + d_F)^{-1}$ &
{$\!\begin{aligned} \sum_{a_\lambda,x} \left(C^{a_\lambda}_{x,\tilde{x}}\right)^2/\operatorname{dim}(\lambda)\end{aligned}$}
\\
\addlinespace 
$s_\lambda f_\lambda(x,U)$ &
{$\!\begin{aligned}  |\langle x|U|\psi_0 \rangle|^2\!\! - \! \tfrac{\tr_n(\rho_0)}{d_F} \end{aligned}$}
& 
{$\!\begin{aligned} \sum_{a_\lambda}C^{a_\lambda}_{x, \tilde{x}} \langle \psi_0|U^\dagger a_\lambda U| \psi_0 \rangle\end{aligned}$}
 \\
\bottomrule
\end{tabular}
\caption{Overview of the range of $\lambda$, normalization factors $n_\lambda$ and $s_\lambda$ and the filter functions for the RAB post-processing for interacting and non-interacting dynamics.
$C^{a_\lambda}_{x,\bar x}$ is the Clebsch-Gordon coefficient expanding the basis vector $a_\lambda$ of $\lambda$ in terms of the Fock basis vectors $x$ and $\bar x = d_F - x + 1$. $\tr_n$ denotes the trace on the irrep $\lambda = n$. 
More details are given in the Supplemental Material \ref{par:filter-functions}.
}
\label{table:summary-protocol-ingredients}
\end{table}

\paragraph*{Protocol specification and analytical guarantees.}  
The Hamiltonians Eq.~\eqref{eq:quartic hamiltonian} generate groups of interacting and non-interacting particle-number preserving quenches. 
As we show, the group action $U\otimes \overline{U}$ restricted to $n$~particles decomposes into the trivial irrep and its complement for interacting, and into $n+1$ irreps for non-interacting quenches (Lemma~\ref{lem:interreps} and \ref{lem:noninterirreps} of App.~\ref{sec:representation-theory-for-particle-number-preserving-dynamics}). 
Table~\ref{table:summary-protocol-ingredients} further summarizes our results (App.~\ref{app:filterRB}) for the filter function and normalization constants. 
The expressions involve Clebsch-Gordon coefficients to construct bases for the irreps of outer products of representations of the unitary groups.

At sufficient sequence length we expect the random dynamics to probe the noise evenly in every inequivalent irrep.
Then, for each one, the average RB signal is described by a single scalar exponential decay. The role of the filter function is to isolate the signal of an individual irrep (using character orthogonality relations). 
Building upon the general statement derived in Ref.~\cite{heinrich2023randomized}, 
we formally show for general \emph{gate-dependent} particle-preserving noise 
that $\hat F_\lambda(m)$ 
averaged over the random sequences and the measurement outcome is well-approximated by Eq.~(\ref{eq: exponential_decay}) whenever $m$ is \emph{larger} than a \emph{minimal admissible sequence length} $m_{\mathrm{ths}}$  (Theorem~\ref{theorem:signal-guarantee-nonint} in the appendix).  
The value of $m_{\mathrm{ths}}$ depends on the measure~$\nu$. 
For non-interacting dynamics, we provide an upper bound on $m_{\mathrm{ths}}$ in terms of the spectral gap $\Delta^{(2n)}_\nu$ of the measure $\nu$ quantifying its deviation from a unitary $2n$-design on the space of single particles.  
We find that choosing $m_{\mathrm{ths}} \gtrsim (\Delta^{(2n)}_\nu)^{-1} l \log d$ suffices for RAB for the irrep with $\lambda = (n,l)$ and $l\in[n]$. 
An analogous result holds for interacting quenches. For gate-independent noise, the decay parameters $z_{\lambda}$ can be further related to the average fidelity of the noise channel (App.~\ref{Ap.:Fitting_model_and_signal_guarantees}). %as shown in

To evaluate the bound on $m_{\mathrm{ths}}$ we further relate (App.~\ref{secapp:convergence-analysis}) the spectral gap to the frame potential \cite{gross_evenly_2007}. 
This allows us to numerically determine $m_{\mathrm{ths}}$ using Monte Carlo estimation of the frame potential. 
These worst-case bounds will typically yield conservatively large estimates for $m_{\mathrm{ths}}$.   
For example, for global depolarizing noise of strength $p$, we can directly calculate the RAB signal to be of the form $\bar F_\lambda(m) = p^m (1 - \alpha(m))$ with an asymptotically vanishing function $\alpha$ independent of $p$ (Lemma~\ref{lem:rab-signal-under-depolarising-noise} in the appendix).  
This motivates us to determine $m_\text{\rm ths}$ from simulations of the RAB protocol without noise $p=0$ such that $|\alpha(m)|$ is sufficiently small. 
We thereby typically find a smaller bound on $m_\text{\rm ths}$ for sufficiently isotropic noise. 

An important aspect of any benchmarking protocol is its complexity in terms of the required number of measurement samples. We obtain an upper bound on the sampling complexity of the protocol for the largest irreducible representation $\lambda = (n,n)$.
We further show that in the \emph{collision-free} regime $d\gg n^2$, the overlap of the measurement operator as well as of the input state is almost exclusively with this irrep.
As we show, it then follows that sample complexity is bounded in $O(d^{2n})$ with $d$ the number of modes (Corollary \ref{cor:collision-free-sample-complexity}).  In particular, the scheme is efficient for fixed $n$ not scaling with $d$.
The detailed derivation is given in App.~\ref{app:subsec:sample-complexity}.

\paragraph*{Numerical simulations.}
We demonstrate our protocol with numerical simulations.
Ref.~\cite{Jadwiga_github} provides a companion open-source Python library for RAB and the presented simulations. 
We consider a system of $d=4$ modes and $n=2$ particles and two different ensembles of nearest-neighbor Hamiltonians. 
In the first ensemble, the hopping strengths, local potentials and onsite interactions are drawn uniformly at random.
The on-site potential is drawn as $h_{i,i} \sim [-1,1]$, whereas the nearest-neighbour values are complex numbers with both the real and imaginary part drawn uniformly from the interval $[-1,1]$.
The on-site interaction, if there is one, is drawn as $V_i \sim [-1,1]$, and the constant time evolution parameter is set to $\Delta t =1$. 
The second ensemble is inspired by the parameter settings used for an analog simulation on a Sycamore chip in Ref.~\cite{hangleiter2021precise}.  
Here the hopping strength is kept constant with only the local potentials being drawn at random, $h_{i,i}\sim [-20,20]~\mathrm{MHz}$, $h_{i,i+1}=-20~\mathrm{MHz}$, and on-site interaction $V_i = -5\mathrm{MHz}$.
The evolution time for the Sycamore setting is $\Delta t=25\mathrm{ns}$.

\begin{figure}
\includegraphics[width=1\linewidth,left]{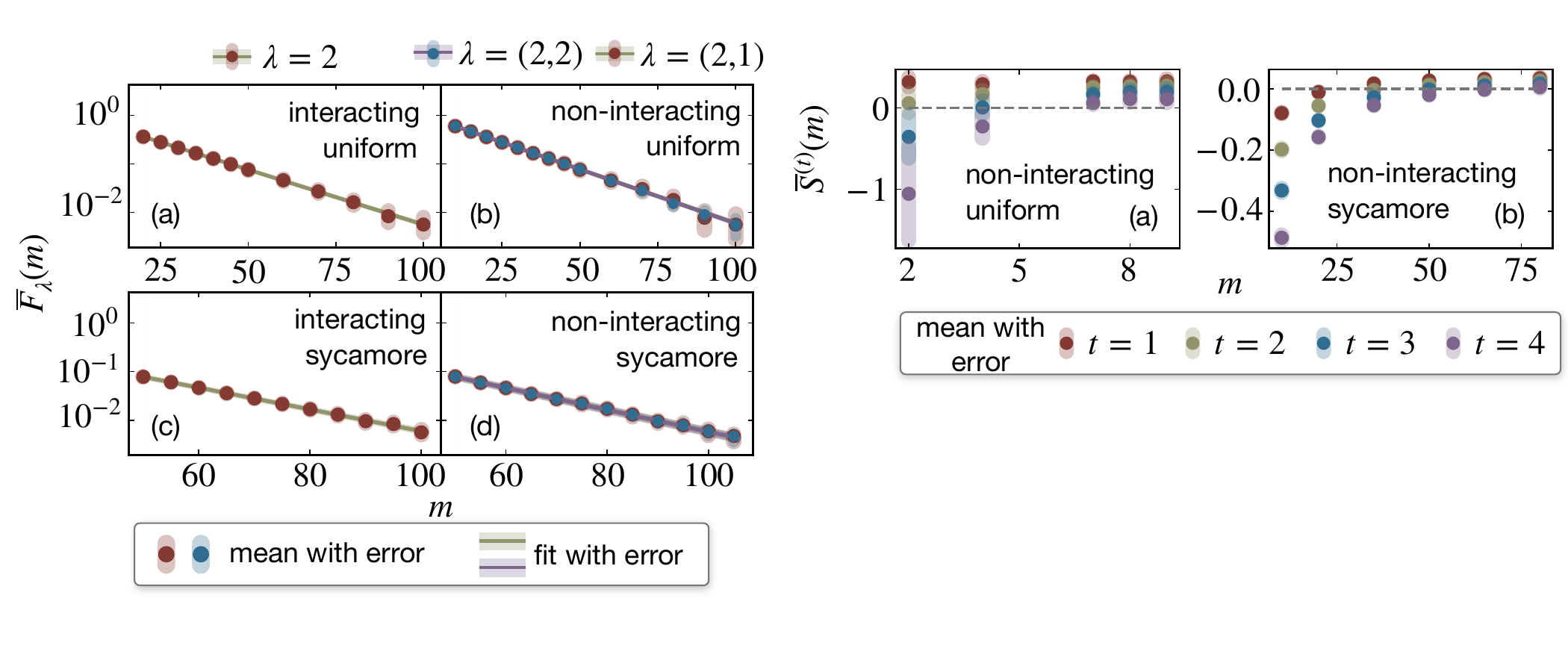}
  \caption{\emph{Numerical lower bound $\bar S^{(t)}(m)$ of the spectral gap} as a function of sequence length from Monte Carlo estimates of the frame potential; once the error bars are above zero, the estimate is sufficient. (a) for the uniform distribution, and (b) for the Sycamore distribution. This informs the analytical warm-up phase.}
  \label{fig:spectralgap-main-text}
\end{figure}

\begin{figure}
\includegraphics[width=1\linewidth, left]{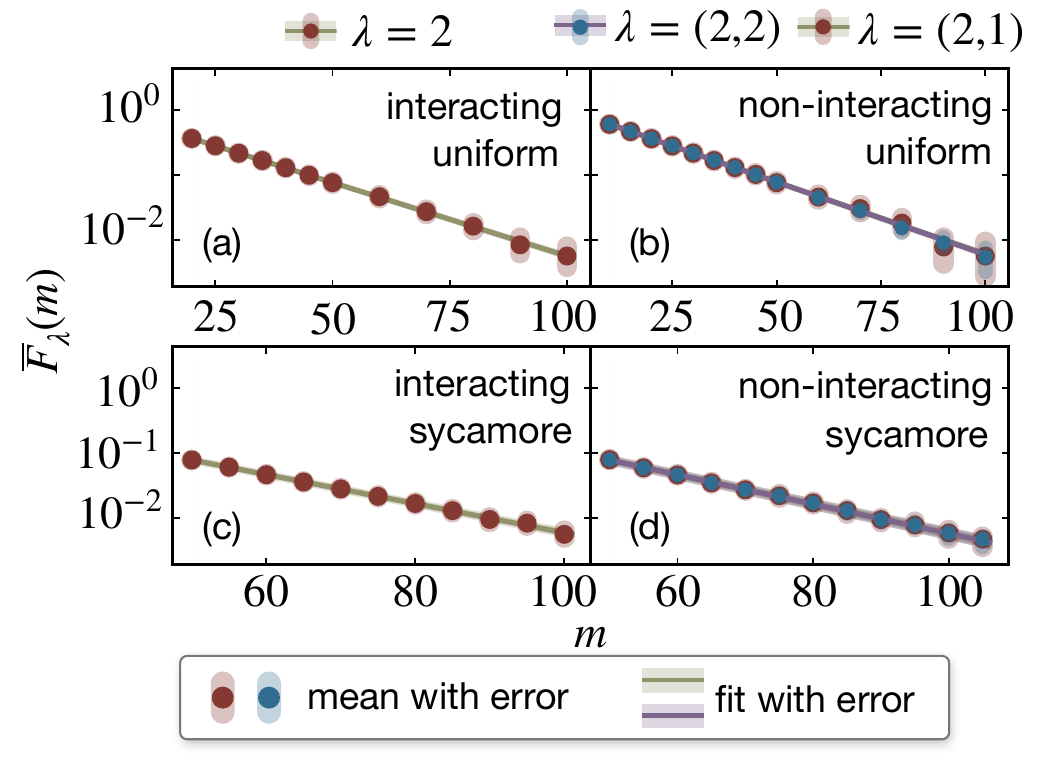}
  \centering
  \caption{\emph{Numerical simulations of RAB protocol.} The left panel shows the decay for interacting dynamics for both the uniform (a) and Sycamore (c) distribution, right panel is the same for non-interacting. Plotting and fitting occur after a warm-up phase. The extracted parameters are $z_2 = 0.949(3)$, $z_{2,1} =0.950(1)$, $z_{2,2} = 0.950(3)$ for the uniform and sycamore distributions are $z_2 = 0.95(0)$, $z_{2,1} =0.949(4)$, $z_{2,2} = 0.94(9)$.}
  \label{fig:rab-results-main-text}
\end{figure}

We begin with evaluating the bounds on the minimal admissible sequence length $m_{\text{ths}}$ for the Hamiltonian ensembles. 
Fig.~\ref{fig:spectralgap-main-text} displays bounds $S_\nu^{(2n)}$ on the spectral gap $\Delta_\nu^{(4)}$ for $n=2$ ($t=4$) using estimates of the frame potential for different sequence length $m$. 
For $n=2$ ($t=4$), we determine an overall lower-bound for $\Delta_\nu^{(4)}$  of $0.11\pm0.05$ for the uniform setting, and for the Sycamore setting $0.05\pm0.03$. From these values we analytically find (using Theorem~\ref{theorem:signal-guarantee-nonint}) that $m_{\mathrm{ths}} =178$ for the uniform setting and $m_{\mathrm{ths}} = 393$ for the Sycamore setting, both values refer to the larger irrep. 
See the App.~\ref{subsec:estimation-of-mths} for details.
The uniform ensemble converges to sufficiently mixing dynamics considerably faster than the Sycamore ensemble. 
However, the minimal sequence lengths obtained in this way are challenging in practice. 
Using simulations of the noiseless RAB signal (App.~\ref{subsec:simulating-warm-up-phase-with-RAB-protocol}), we determine considerable lower admissible values for $m_\text{\rm ths}$ ($16$, $8$, and $55$ for the interacting ensemble, non-interacting uniform, and Sycamore ensembles, respectively). 
These values seem to be sufficient in practice also for the noise models considered in the following. 

For simulating RAB experiments, we model the data acquisition with a gate-independent depolarizing noise channel $\Lambda_p=(1-p)\rho+p\frac{\mathbb{I}}{d_F}$ acting after every unitary quench on the relevant Fock-sector, with average depolarizing strength $p=0.05$. The relative frequencies are computed using $N=10$ shots for the uniform setting and $N=100$ for the Sycamore setting and the number of random sequences per sequence length was chosen to be $K=10^5$.

In Fig.~\ref{fig:rab-results-main-text} we present the RAB results for both of our Hamiltonian ensembles, the upper quadrant of the uniform ensemble and the lower quadrant of the Sycamore ensemble. 
The data indeed follows the expected decay curve.
Theoretically, in this simple setting, the decay parameters $z_\lambda$ should be approximately equal to $1-p$ and we find that the extracted parameters are compatible with this, up to relative errors below $0.1\%$ for both settings. 
This simple noise model illustrates the protocol functioning, we present further simulation results 
for different noise models (spurious interactions, 
miscalibrated interactions and errors in the quench duration) in App.~\ref{app:numerical-simulations} and consistently observe the expected exponential decay. 
Importantly, we also find that one can, in principle, discern different error models as they exhibit distinct characteristics in the different decay parameters of the irreps. We leave further study of the different RAB signals for different noise mechanisms of specific platforms to future work.

\paragraph*{Discussion.}
In this work, we have presented and analyzed an experimentally feasible benchmarking protocol (RAB) for well-controlled quantum devices implementing particle-number preserving dynamics in the `low-complexity' regime. 
The RAB protocol directly benchmarks the native operations of these systems, namely the Hamiltonian evolution they can implement, rather than relying on implementing compiled gates such as Clifford gates. 
Furthermore, in contrast to previous randomized benchmarking schemes for analog quantum simulators,  
it does not require the inversion of the random evolution before measurement; this is particularly important in this context, where such inverses might not be easily implementable. 
In contrast to digital high-level  benchmarks such as cross-entropy \cite{boixo_characterizing_2016} or the \emph{quantum unitary evolution score} \cite{Dong2022}, our method does \emph{not} aim at benchmarking or verifying the  computational task solved by the device, but should rather be understood as part of the toolbox for diagnosing and calibrating the native operations of the device. 

To theoretically establish the scalability of the RAB protocol, we have bounded its sample complexity for collision-free inputs.  
We expect that with more fine-grained control over the variance of the relevant estimators, one can significantly improve our bound.  
In addition, fine-tuning the choices of initial state and measurement for extracting specific noise effects can further improve the protocol's scalability.

Using our filter functions, along with our numerical implementation
the RAB protocol can readily be implemented on experimental devices.
Specific experimental scenarios to which our protocol applies include atomic bosonic \cite{roushan_spectroscopic_2017} and fermionic analog simulators \cite{xu_frustration-_2023}, subuniversal photonic devices such as Gaussian boson samplers
\cite{hamilton_gaussian_2017,GaussianBosonSamplingPan,madsen_quantum_2022}, as well as universal quantum processors based on integrated optical bosonic systems \cite{RevModPhys.79.135,Fusion,Taballione_2021}, 
or quantum computers based on fermionic cold atoms 
\cite{BravyiFermionic,FermionicComputing}. 
In our work, we have demonstrated the efficacy of our protocol for state-of-the-art system sizes \cite{Taballione_2021}.
We are thus confident that---like RB in the digital setting---RAB will become an important tool for assessing and calibrating bosonic and fermionic quantum devices.
 
\paragraph*{Parallel work.} 
Independently of this work, Arienzo et al.~\cite{parallel}  developed a benchmarking protocol for passive bosonic transformations. 
Their passive bosonic RB (PRB) protocol is also based on filtered RB \cite{heinrich2023randomized} and uses the same post-processing stage and representation theoretic results as our protocol for non-interacting bosonic quenches.  
In contrast to our focus on platform native quenches, PRB uses Haar-randomly sampled unitaries. Arienzo et al.\ also derive and evaluate bounds on the variance of the PRB signal hinting at a sample-complexity scaling logarithmically in the number of modes 
and discuss extensions to Gaussian measurements.

\paragraph*{Acknowledgements.}
We thank Jonas Helsen,  Michal Oszmaniec and Markus Heinrich for the discussions. 
We acknowledge support from the BMBF 
(project DAQC, for which it provides tools for randomized benchmarking of digital-analog 
architectures, and project FermiQP and MuniQC-Atoms, for benchmarking cold atomic analog many-body devices), the Munich Quantum Valley (K-8), the quantum Flagship (PasQuans2, Millenion), and the ERC (DebuQC). We also acknowledge funding by the DFG (CRC 183, EI 519/21-1), the Harvard Quantum Initiative (HQI).
JW acknowledges financial support from the SFB BeyondC of the Austrian Science Fund (FWF) [Grant-DOI:Grant-DOI 10.55776/F71] (SFB BeyondC) and the HPQC flagship project of the Austrian Research Promotion Agency (FFG).
DH acknowledges funding from the US Department of Defense through a QuICS Hartree fellowship.
 
\ifjournal
    \putbib
    \end{bibunit}

    \begin{bibunit}
\fi

\newcommand{\nocontentsline}[3]{}

\let\tmp\addcontentsline
\let\addcontentsline\nocontentsline
% \ifarxiv
    \bibliographystyle{./myapsrev4-1}
    \bibliography{bib_analogue_rb,hamiltonian_learning}
% \fi

\let\addcontentsline\tmp

\onecolumngrid
\newpage
\cleardoublepage
\setcounter{page}{1}
\setcounter{equation}{0}
\setcounter{footnote}{0}
\setcounter{figure}{0}
\thispagestyle{empty}

\appendix 
\begin{center}
\textbf{\large Supplementary material for ``Benchmarking bosonic and fermionic dynamics''}\\
\vspace{2ex}
J. Wilkens, M. Ioannou, E. Derbyshire, J. Eisert, D. Hangleiter, I. Roth, J. Haferkamp
\vspace{2ex}
\end{center}

\appendixtableofcontents
% \newpage

\section{Preliminaries}
\label{app:preliminaries}

\subsection{Notation} 

For $A, B$ linear operators on a vector space, we denote the Hilbert-Schmidt inner product by $(A|B) = \tr[A^\dagger B]$ and will make extensive use of standard Dirac (bra-ket) notation with respect to the Hilbert-Schmidt inner product. 
Switching between the corresponding canonical vector space isomorphism without saying, we will understand that $A\otimes B|C) = | A C B^T)$ for $A,B,C$ linear operators, with $(\cdot)^T$ denoting transposition. 
The space of linear operators or endomorphisms on a vector space $V$ will be denoted by $L(V)$. 
$L(V)$ is isomorphic to the space $V^{\otimes 2}$, via the canonical vector space isomorphism explained above.
In particular, we will repeatedly encounter the space $L(\mathrm{Sym}^n(\mathbb C^d))$ of linear operators on the symmetric subspace of $(\mathbb C^d)^{\otimes n}$ (where the permutation group acts via commuting the tensor factors). 

We use $\|A\|_\infty = \max \operatorname{spec}(A)$ for the operator norm of $A$, i.e., the maximal singular value. Further, we denote the Frobenius-norm (a.k.a. Schatten-$2$ norm) by $\norm{A}_F = [\tr(A^\dagger A)]^{\tfrac12}$ and the nuclear norm (a.k.a. trace-norm, Schatten-$1$ norm) by $\norm{A}_1 = \tr\sqrt{A^\dagger A}$.

\subsection{Representation designs}
\label{subsec:representation-designs}
An important aspect of any randomized benchmarking scheme is that the random ensemble of unitary operations one benchmarks is `mixing' some subspaces. 
These subspaces are the irreducible representations (spaces) of the group that are generated by the benchmarked operations. 
To generally capture the notion of randomly mixing the subspaces, we introduce the notion of a \emph{design} for any representation $\rho$ of a compact group $G$ with (normalized) Haar measure $\mu$. 
Let $M^{\rho}_\nu = \int_{G} \rho(g)\, \rm d\nu(g)$ be the $\rho$-moment operator of a measure $\nu$. 
We call a measure $\nu$ on $G$ a \emph{$\rho$-design} if 
\begin{equation}
    M^\rho_\nu = M^\rho_\mu\,.  
\end{equation}
Correspondingly, we say $\nu$ is an $\epsilon$-approximate $\rho$-design with respect to a norm $\| \cdot \|$ if $\|M^\rho_\nu - M^\rho_\mu \| \leq \epsilon$.
Another useful object, when studying the convergence of sequences of operations drawn from a measure $\nu$, is the spectral gap $\Delta^\rho_\nu$. The spectral gap quantifies the distance between the largest and second largest eigenvalues of the moment operator as
\begin{equation}\label{eq:spectralgap}
    \|M_\nu^\rho - M_\mu^\rho\|_\infty = 1 - \Delta_\nu^\rho.
\end{equation}
In particular, the product $\rho(g_m) \cdots \rho(g_1)$ with $m$ elements of $g_1$, \ldots , $g_m$ i.i.d.\ drawn according to $\nu$ is a random matrix $\rho(g)$ with $g$ drawn according to the $m$-fold convolution $\nu^{\ast m}$ of $\nu$. And we have $\Delta^\rho_{\nu^{\ast m}} = \left(\Delta^\rho_\nu\right)^m$. 
This makes spectral gaps a convenient tool to study the convergence of random sequences to approximate designs. 

In the context of quantum information most prominent are \emph{unitary $t$-designs} \cite{gross_evenly_2007}. 
A unitary $t$-design is a measure $\nu$ on $\rm U(d)$ that is a $U^{\otimes t} \otimes \bar U^{\otimes t}$-design. 
 Let us simply denote the \emph{$t$-th moment operator} as
\begin{equation}\label{eq:moment_operator}
    M_\nu^{(t)} = \int_{U(d)} U^{\otimes t} \otimes (U^\dagger)^{\otimes t}d\nu (U) 
\end{equation}
for unitaries drawn from a measure $\nu$ and the corresponding spectral gap as $\Delta^{(t)}_\nu$.

Generally speaking, filtered RB protocols require $\rho^{\otimes 2}$-designs of the representation of the group that one is benchmarking \cite{heinrich2023randomized}.  
As we will show, for the case of particle-number preserving dynamics, however, a simpler (slightly stronger) sufficient condition can be formulated in terms of approximate unitary $t$-designs of high-enough order $t$ on specific subspaces.

By choosing different norms and moment operators, there are several distinct ways to quantify the distance between a distribution of unitaries and the Haar measure over the group. 
Yet another way to measure the distance is the deviation of the so-called \emph{frame potential} from its minimal value \cite{gross_evenly_2007}. 
The frame potential is the double average taken over an ensemble of unitaries as
\begin{equation}\label{eq: frame_potential}
    \mathcal{F}_\nu^{(t)} = \int_{U(d)}\left|\mathrm{tr}\left[UV^\dagger\right] \right|^{2t}  \mathrm d\nu(U) \mathrm d\nu(V)\, 
\end{equation}
and its minimal value of $t!$ is attained for the Haar measure $\mu$.

\newpage
\section{Representation theory for particle-number preserving dynamics}
\label{sec:representation-theory-for-particle-number-preserving-dynamics}

\begin{figure}[tb]
    \centering
    \includegraphics[width=\textwidth]{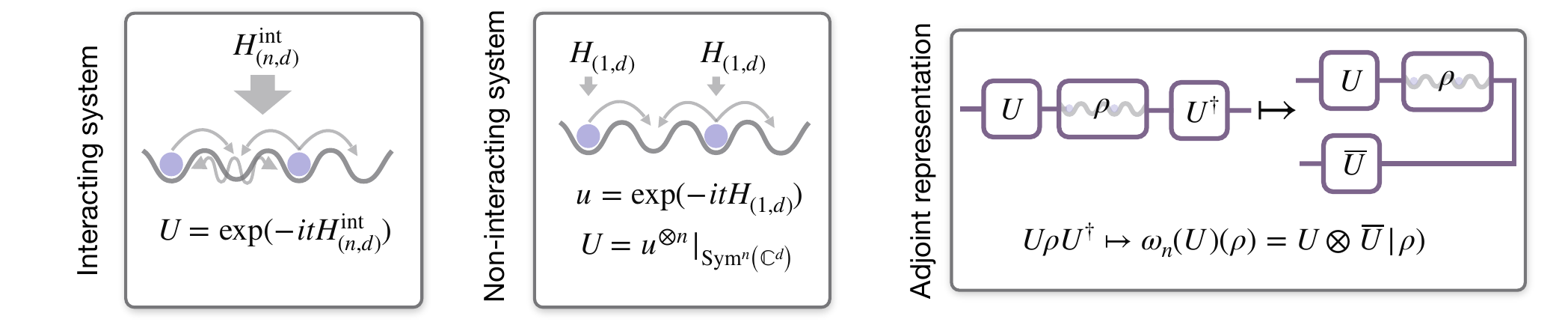}
    \caption{(left panel) Depiction of interacting dynamics, the whole system can only be described by a single Hamiltonian $H^{\mathrm{int}}_{(n,d)}$ acting on the whole Fock space.
    (middle panel) Depiction of non-interacting dynamics, each particle evolves under the same Hamiltonian $H_{(1,d)}$ defined on the single particle Fock space.
    (right panel) Schematic of the conjugation representation in tensor networks, where $\rho$ is a quantum state and the over line $\overline{\cdot}$ means complex conjugation. 
    }
    \label{fig:hamiltonian-int-nonint-rep}
\end{figure}

In this work, we take a bottom-up approach to benchmark analog quantum devices. 
We start by taking potentially available `native' operations of the device at face value and design the RB protocol for these operations.
To this end, we have to understand the subspaces that the RB sequences will `mix' and how fast this will happen. 
The first question is naturally answered in the language of representation theory, as such subspaces carry irreducible representations. 
In the following sections, we will derive the corresponding irreducible representation for groups generated by particle-number preserving Hamiltonians with and without interaction terms, respectively. 
We further construct the projectors onto these irreducible representations and describe their dimensions.
These are the defining quantities for our analog benchmarking protocol and are summarized in Table~\ref{table:summary-protocol-ingredients} in the main text.
A summary of the two Hamiltonian ensembles can be found in Fig.~\ref{fig:hamiltonian-int-nonint-rep}.

The Hilbert space describing the configuration space of a bosonic system of particles in $d$ modes is the Fock space 
    $\mathcal F = \overline{\bigoplus_n \Sym^n (\mathbb C^d)}$, 
where every Fock sector in the direct sum describes configurations of $n$ particles.  
(Formally, the Fock space is the closure of infinitely many Fock sectors. 
However, for our purposes it will be sufficient to consider the space restricted to a finite number of particles allowing us to treat the Fock space as finite dimensional.) 
The dimension of the $n$-th Fock sector is $d_F(n) \coloneqq \dim \Sym^n(\mathbb C^d) = {n+d-1\choose n}$.
We will often drop the dependence on $n$ and simply write $d_F$ if $n$ is clear from the context.
Moreover, we denote by $P_{\mathrm{Sym}^n}$ the projector onto the $n$-th Fock sector.
In terms of the bosonic creation operators $a_j^\dagger$ appearing in Eq.~\eqref{eq:quartic hamiltonian} a convenient basis is constructed by symmetrizing states of the form $a_{j_n}^\dagger \ldots {a^\dagger_{j_1}} \ket 0$, $j_k \in \{1, \ldots d\}$, where a string of $n$ creation operators acts on the vacuum state $\ket 0 \in \Sym^0(\mathbb C^d)$.
We enumerate the states in lexicographical order of decreasing strings of mode labels $j_n \geq \ldots \geq j_2 \geq j_1$ and denote the Fock basis in this order by $\{\ket a  \mid a \in \{0, \ldots, d_F\}\}$.

Since the Hamiltonians of the form Eq.~\eqref{eq:quartic hamiltonian} are 
particle-number preserving, they generate groups of unitaries that act block-diagonally on the Fock space. Such particle-number preserving unitaries $U$ act on density matrices in $L(\mathcal F)$ by conjugation, i.e., as unitary 
quantum channels, 
\begin{equation}
\omega(U) \coloneqq U \otimes \bar U = \bigoplus_n U|_{\Sym^n(\mathbb C^d)} \otimes \bar U|_{\Sym^n(\mathbb C^d)} =: \bigoplus_n \omega_n(U)
\end{equation}
within every Fock sector. 

For the two groups we are considering $\omega_n$ is further reducible for all $n > 0$, i.e., there exists a basis transformation for $L(\Sym^n(\mathbb C^d))$ such that all unitaries within the group are block-diagonal (and have more than one block). 
Concretely, we will find the following decompositions in the remainder of this section:

First, we have the unitaries that are generated by Hamiltonians \emph{carrying interaction terms} (Fig.~\ref{fig:hamiltonian-int-nonint-rep}, left) which means that the evolution operator automatically acts on the full $n$-particle preserving Fock space with $d$ modes. In this case, the representation $\omega_n = \sigma_\text{\rm trivial} \oplus \sigma_n$ decomposes into two irreducible representations: the trivial representation carried by the span of the identity in $L(\Sym^n(\mathbb C^d))$ and its complement $\sigma_n$ carried by trace-less matrices. We denote the carrier space of $\sigma_n$ by $W_n$ and its projector as $P_n$.

Second, the Hamiltonians \emph{without interaction (quartic) terms}, (Fig.~\ref{fig:hamiltonian-int-nonint-rep}, middle), generate a smaller group isomorphic to $\operatorname{SU}(d)$, acting as (anti-)symmetric for bosons (\,/\,fermions) tensor products on the Fock sectors. Now, $\sigma_n = \bigoplus_{l \in \{1, \ldots,  n\}} \sigma_{n,l}$ further decomposes into $n$ mutually inequivalent irreps $\sigma_{n,l}$ with $l \in \{1, \ldots, n\}$ acting on subspace $
W_{n,l}$. 
We denote the projectors onto $W_{n,l}$ by $P_{n,l}$.
See Fig.~\ref{fig:repr_definition} for an illustration and overview. 

\begin{figure}[tb]
     \includegraphics[width=.5\textwidth]{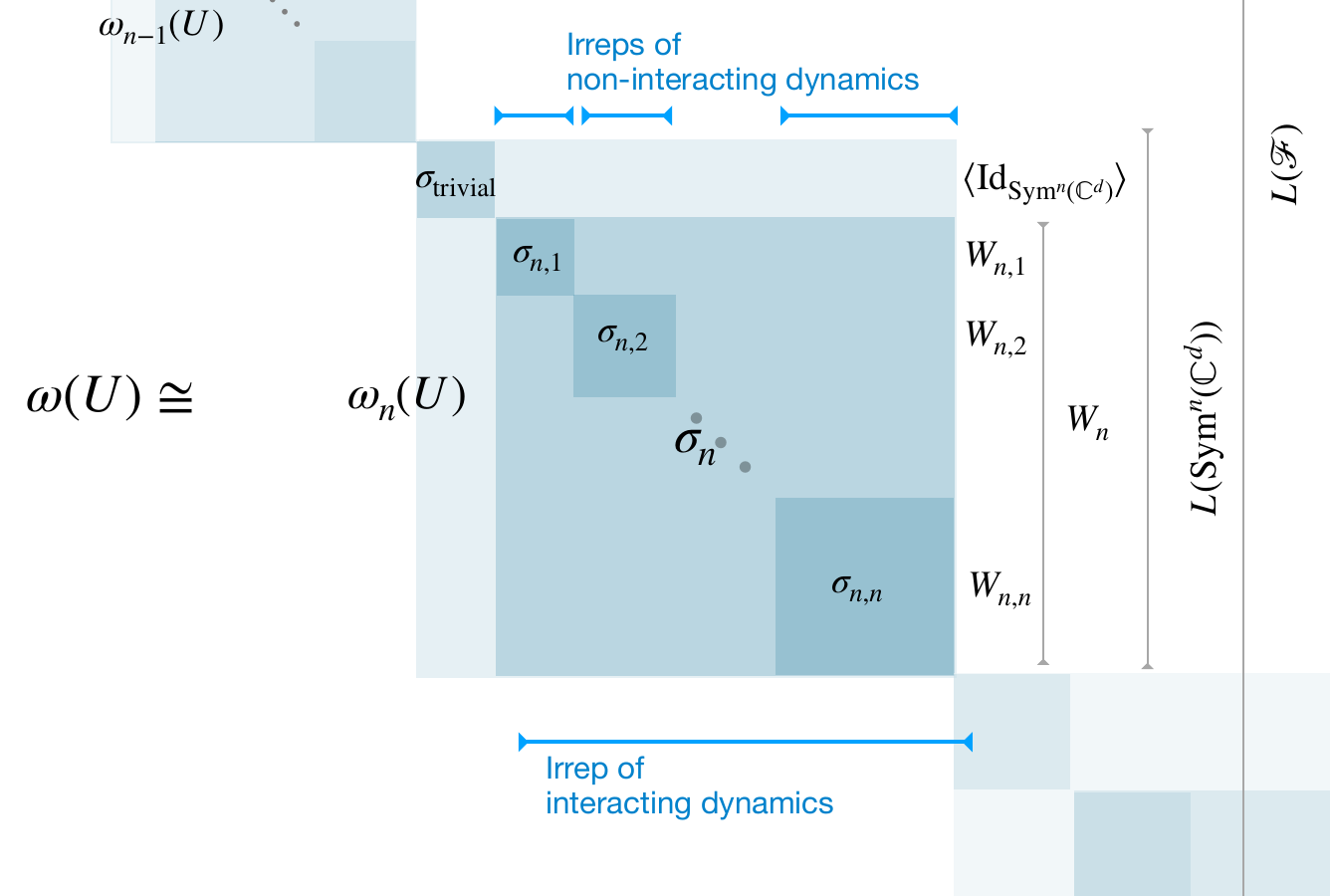}
     \caption{\label{fig:repr_definition}Overview of the notation for the irreducible representations of the groups generated by non-interacting and interacting particle-number preserving bosonic Hamiltonians. For interacting dynamics, the action $\omega(U)$ as a unitary quantum channel on Fock space endomorphisms $L(\mathcal F)$ decomposes into one non-trivial irrep $\sigma_n$ carried by $W_n$ per Fock sector with $n$ particles $\Sym^n(\mathbb C^d)$. In the absence of interactions $\sigma_n$ is further reducible to $n$ representations $\sigma_{n,l}$ carried by spaces $W_{n,l}$ with $l \in \{1, \ldots, n\}$. }
\end{figure}

\subsection{Bosonic, particle-number preserving non-interacting dynamics}

\subsubsection{Irreducible Representation}

The group of unitaries generated by particle-number preserving, quadratic non-interacting Hamiltonians is isomorphic to $\operatorname{SU}(d)$.   A unitary $u \in \operatorname{SU}(d)$ acts faithfully on a Fock sector $\Sym^n(\mathbb C^d)$ via its $n$-th fold symmetric tensor power $\Sym^n(u) = u^{\otimes n} |_{\Sym^n(\mathbb C^d)}$.  In other words, $u$ acts identically and independently on the mode spaces of the particles. In particular, $\Sym^1(u) = u$ and the unitaries, thus, are already completely determined by their action on the single particle sector. 
So for the group of (particle-number preserving) non-interacting unitaries we understand the representation $\omega_n$ as the representation $\omega_n = \Sym^n \otimes\ \overline{\Sym^n}$ of $\operatorname{SU}(d)$ on $L(\Sym^n(\mathbb C^d))$. 
The representation $\Sym^n$ and its complex conjugate $\overline{\Sym^n}$ are irreducible. 
For the corresponding representation of $\operatorname{SU}(d)$ as unitary quantum channels, we take the tensor product of these two irreducible representations that decompose into irreducible representations as follows: 
\begin{lemma}[Reduction of $\omega_n$ for non-interacting unitaries]\label{lem:noninterirreps}
	The space $L(\mathrm{Sym}^n(\mathbb{C}^d))$, decomposes under the action 
    $\Sym^n \otimes\ \overline{\Sym^n}$ of $\operatorname{SU}(d)$ (which is the action of $\omega_n$ representing the group of particle-number preserving non-interacting unitaries), into a trivial representation and $n$ non-trivial mutually inequivalent irreducible representations denoted by $\sigma_{n,l}$ for $l\in{0,\ldots, n}$.  The carrier space $W_{n,l}$ of $\sigma_{n,l}$ has dimension 
    \begin{equation}\label{eq:dimensionformulabosons}
        \dim W_{n,l}= \frac{d+2l-1}{d-1}{d+l-2 \choose l}^2={d+l-1\choose l}^2-{d+l-2\choose l-1}^2\,.
    \end{equation}
\end{lemma}

\begin{proof}
 The irreducible representations of $\mathrm{SU}(d)$ are in one-to-one correspondence to so-called Young diagrams---a graphical representation of partitions of natural numbers into summands. 
	We can determine the irreducible representations of $\Sym^n \otimes\ \overline{\Sym^n}$ using the calculus of Young diagrams, see e.g.~Ref.~\cite{georgi_lie_1999} for an explanation of the technique.
    The bosonic representation $ \Sym^n $ corresponds to the Young diagram
{\scriptsize\[
	\underbrace{\ytableausetup{centertableaux}
		\begin{ytableau}
		~&  &  & \cdots &  &  &
		\end{ytableau}}_{n}\,.
	\]}
	Using standard rules, we obtain the following Young table for the conjugate representation $u \mapsto \overline{\Sym^n}(u) = \overline{u}^{\otimes n}|_{\Sym^n(\mathbb C^d)}$: 
	{\scriptsize\[
	\ytableausetup{centertableaux}
	\left.	\begin{ytableau}
	~&  &  &\cdots &  &  & \\
	~&  &  &\cdots &  &  & \\
	~&  & \vdots &\vdots & \vdots &  & \\
	~&  &  &\cdots &  &  & \\
	~&  &  &\cdots &  &  & 
	\end{ytableau}~\right\}d-1\,, 
	\]}
    which consists of $d-1$ rows of $n$ boxes each.  Applying the rule for decomposing tensor products of two representations, we 
    find a decomposition of 
    $\Sym^n \otimes\, \overline{\Sym^n}$ into $n+1$ irreducible representation. 
	For convenience and illustration, we write out explicitly the decomposition for $n=4$ and $d=4$, 
	\begin{multline*}\scriptsize
	\ytableausetup{centertableaux}
	\begin{ytableau}
	~&   &  &
	\end{ytableau}\quad\otimes \quad
	\begin{ytableau}
	~&    &  & \\
	~&    &  & \\
	~&   &  & 
	\end{ytableau}\quad =
	\quad
	\begin{ytableau}
	~&    &  & & & & &\\
	~&    &  &  & \none & \none &\none &\none \\
	~~&    &  && \none & \none &\none &\none \\
	\end{ytableau}
	\quad\oplus\quad
	\begin{ytableau}
	~  &  & & & & &\none\\
	~   &  &  & \none & \none &\none &\none \\
	~    &  & & \none & \none &\none &\none 
	\end{ytableau}
	\oplus\quad
	\begin{ytableau}
	~  & & & \\
	~ &  & \none & \none \\
	~~ & & \none & \none  
	\end{ytableau}\quad \oplus\quad
	\begin{ytableau}
	~ & & \none \\
	~ & \none & \none   \\
	~ & \none & \none  \\
	\end{ytableau}\quad\oplus\quad
\mathrm{Triv}\,.
	\end{multline*}
  The diagrams on the right-hand side arise from concatenating part of the boxes of $\Sym^n$ to the first row of the diagram of $\overline{\Sym^n}$ and the remaining boxes as another row.
    Then, we can delete all columns with $d$ boxes. In particular, the last diagram corresponds to the empty diagram and hence denotes the trivial representation.
    Following this strategy for general $n$ and $d$, we can add $l \in {0, \ldots, n}$ boxes of $\Sym^n$ to the first row of $\overline{\Sym^n}$.  In this way, we arrive at all diagrams of the following form: $2l$ boxes in the first row and $d-1$ subsequent rows of $l$ boxes for $l \in \{0, \ldots, n\}$. 
    For $l=0$  we always encounter a trivial irrep. For the remaining inequivalent ones, we introduce the short label $\lambda = (n,l)$. 
    The stated dimension of the irreps directly follows from general formulas given their Young diagrams (compare~\cite[Chapter 12.2]{georgi_lie_1999}). 

\end{proof}

\subsubsection{Projectors and their construction}
\label{subsubsec:nonint-projectors-and-their-construction}

In the following, we characterize the projectors on the irreducible representations.
Roughly, the $(n+1)$th irreducible representation can be viewed as the complement of all the linear maps that act only on $n-1$ many particles.
This leads to a recursive formula for the projectors onto the irreducible representations $W_{n,l}$.

\begin{theorem}\label{theorem:explicitdescriptionofirreps}
Consider the space of linear maps acting on $r\leq n$ bosonic particles defined as
\begin{equation}
 L_r\coloneqq \{P_{\mathrm{Sym}^n}(\mathbbm{1}_{d}^{\otimes n-r}\otimes M)P_{\mathrm{Sym}^n}, M\in L(\mathrm{Sym}^r(\mathbb{C}^d))\}\,.
\end{equation}
 $L_r$ is isomorphic as a subrepresentation to the space $L(\mathrm{Sym}^r(\mathbb{C}^d))$ equipped with the action $\omega_r$.
Moreover, the projectors $P_{n,l}$ satisfy
\begin{equation}
    P_{n,l}\coloneqq P_{L_l}-P_{L_{l-1}}\, .
\end{equation}
\end{theorem}
Theorem~\ref{theorem:explicitdescriptionofirreps} states that the space $L(\text{Sym}^n(\mathbb{C}^d))$ is structured like an onion with $n+1$ rings, where each layer is an irreducible representation and $L_r$ contains the core (the trivial irreducible representation) and the $r$ most inner rings. 
\begin{proof}
We claim that the map $i:\mathrm{Sym}^r(\mathbb{C}^d)\to L_r, M\mapsto P_{\mathrm{Sym}^n}(\mathbbm{1}_{d}^{\otimes n-r}\otimes M)P_{\mathrm{Sym}^n}$ is an isomorphism of representations.
We find
\begin{align}
\begin{split}
    i(U^{\otimes r}M(U^{\dagger})^{\otimes r})&= P_{\mathrm{Sym}^n}(\mathbbm{1}_{d}^{\otimes n-r}\otimes U^{\otimes r} M(U^{\dagger})^{\otimes r})P_{\mathrm{Sym}^n}\\
    &= P_{\mathrm{Sym}^n}U^{\otimes r}(\mathbbm{1}_{d}^{\otimes n-r}\otimes M)(U^{\dagger})^{\otimes r}P_{\mathrm{Sym}^n}\\
      &=U^{\otimes n} P_{\mathrm{Sym}^n}(\mathbbm{1}_{d}^{\otimes n-r}\otimes M)P_{\mathrm{Sym}^n}(U^{\dagger})^{\otimes n}\\
      &=U^{\otimes n}i(M)(U^{\dagger})^{\otimes n}.
      \end{split}
\end{align}
Therefore, $i$ is an intertwiner
and $L_r=i(L(\mathrm{Sym}^r(\mathbb{C}^d)))$ is a subrepresentation of $L(\mathrm{Sym}^n(\mathbb{C}^d))$.

Next, we need to prove that $i:\mathrm{Sym}^r(\mathbb{C}^d)\to L_r$ is an isomorphism.
For this, we show that there is a matrix in $L_r$ that is not in $L_{r-1}$. 
More precisely, we claim that the matrix $P_{\mathrm{Sym}^n} \mathbbm{1}_d^{n-r}\otimes |1\rangle_r\langle {r+1}|_r\, P_{\mathrm{Sym}^n}$ for $d\geq 2$ with $\ket {k}_r$ the $k$th Fock basis state of $\mathrm{Sym}^r(\mathbb C^d)$ is in $L_r$ but not in $L_{r-1}$. 
Recall that $|1\rangle_r$ and $|r + 1\rangle_r$ are the state with all $r$ particles in the first and second mode, respectively. 
We find
\begin{equation}
    P_{\mathrm{Sym}^n} \mathbbm{1}_d^{n-r}\otimes |1\rangle_r\langle r + 1|_r P_{\mathrm{Sym}^n}=\sum_{j_1,\ldots, j_{n-r}=1}^d P_{\mathrm{Sym}^n}\bigotimes_{i=1}^{n-r}|j_i\rangle\otimes |1\rangle_r \langle j_i|\otimes \langle r+1|_r P_{\mathrm{Sym}^n}.
\end{equation}
In particular, for each state in this sum, the difference in the number of particles in mode 1 in the kets vs bras is at least $r$, which is orthogonal to all matrices in $L_{r-1}$.
As  $L_0\subseteq\ldots\subseteq  L_r\subseteq\ldots \subseteq L_n$ and $L_r\neq L_{r+1}$, we find a contradiction unless $L_r$ contains exactly $r+1$ of the irreducible representations  $W_{n,r}$.

Finally, we need to determine which representations lie in $L_r$. 
Notice that $\dim L_r\leq {d+r-1\choose r}^2$ and therefore, we find for the orthogonal complement of $L_{r-1}$ in $L_{r}$ that 
\begin{equation}\label{eq:irreducible repsdimensioninequality}
\dim L_{r-1}^{\perp}\geq {d+r-1\choose r}^2-{d+r-2\choose r-1}^2.
\end{equation}
We compare the above inequality with the dimension formula Eq.~\eqref{eq:dimensionformulabosons} to match the spaces $L_r$ and $W_{n,l}$: It is easy to see that the only one-dimensional spaces are $L_0=W_{n,0}$. 
Moreover, except for $W_{n,0}$ only $W_{n,1}$ can satisfy Eq.~\eqref{eq:irreducible repsdimensioninequality} for $r=1$ as the irreducible representation dimensions grow monotonically in $r$.
Repeating this argument inductively shows $L_{r-1}^{\perp}=W_{n,r}$.
\end{proof}

As is customary for RB schemes, we must (classically) simulate the ideal action of the sequence of operations within a given irrep. 
For standard Clifford RB the classical simulation is required for determining the sequence inverse that only acts non-trivially on the subspace of trace-less matrices.
To restrict the simulation to a particular irrep we require expressions for the projectors onto the carrying spaces. 
These will later enter into the construction of `filter functions' for the protocol.

Clebsch-Gordan coefficients describe the corresponding transformation from a basis of the tensor product of two irreps $W_{\lambda} \otimes W_{\lambda'} \cong \bigoplus_{\lambda''} W_{\lambda''}$ to the bases of the irreps $W_{\lambda''}$ appearing in its complete decomposition, see e.g.~Ref.~\cite{alex2011numerical}. 
Therefore, we can explicitly construct an expression of the projectors $P_{n,l}$ in terms of the basis of $L(\Sym^n(\mathbb C^d)) \cong\ \Sym^n(\mathbb C^d) \otimes \Sym^n(\mathbb C^d)$ from the corresponding Clebsch-Gordan coefficients.
Since here $\Sym^n(\mathbb C^d) \otimes \Sym^n(\mathbb C^d)$ carries the unitary representation $\Sym^n \otimes\ \overline{\Sym^n}$, a natural choice of basis for the second tensor factor is a `conjugate basis' of the Fock basis.
This basis arises from adding a minus sign to the Lie algebra and thereby changing the role of highest and lowest weight states in the construction.  
To this end, we denote $\bar a = d_F - a - 1$ for indices $a \in \{0, \ldots ,d_F-1\}$ and further understand that $\ket{a} = \sigma_{a} \ket{\bar{a}}_F$ with $\sigma_a \in \{-1, +1\}$ the signs arising in the conjugate construction is the basis for the second tensor factor and $\ket{a}_F$ the `original' Fock basis. 
Together $\{\ket{a, a'} = \ket a \ket{a'} = \sigma_{a'} \ket a \ket {d_F - a' - 1}_F \mid a, a' \in {0, \ldots, d_F}\}$ is our basis for $\Sym^n(\mathbb C^d) \otimes \Sym^n(\mathbb C^d)$.

In terms of Clebsch-Gordan coefficients $C^{a_\lambda''}_{a, a'}$ basis vectors for the irrep $\lambda$ arising in the decomposition of $\Sym^n \otimes\ \overline{\Sym^n}$ can be constructed as 
\begin{equation}
    |a_\lambda''\rangle = \sum_{a,a'} C^{a_\lambda''}_{a, a'} |a, a'\rangle\,,
\end{equation}
where $a, a' \in \{0, \ldots, d_F - 1\}$ and $a_\lambda'' \in \{0, \ldots, \dim(W_\lambda) - 1\}$.

The projector onto the irreducible representation $W_\lambda$ then reads
\begin{align}
    P_\lambda = \sum_{a_\lambda''}|a_\lambda''\rangle\langle a_\lambda''|
    = \sum_{a''_\lambda} \sum_{a,a',b, b'} C^{a''_\lambda}_{b,b'} C^{a''_\lambda}_{a,a'} |a, a'\rangle\langle b, b'|
    \,,
    \label{eq:projector-passive-gaussian}
\end{align}
where the first sum runs over the basis elements of $\lambda$. 

In Ref.~\cite{alex2011numerical} an algorithm for Clebsch-Gordan coefficients is presented and implemented in the C++ programming language.  We will use the reference's implementation for numerical evaluations of expressions involving Clebsch-Gordan coefficients.

\subsubsection{Overlap with irreducible representations}\label{app: overlap}
The RB protocol probes the implementation of a set of operations per irrep using an initial state. 
In order to succeed it is important to ensure sufficient overlap of initial states (and measurement) with the irreps under scrutiny.  
To this end, the following structural result will be useful in order to obtain bounds on the sample complexity. 

\begin{lemma}\label{lower_bound_overlap}
For $|\psi\rangle\in \mathrm{Sym}^n(\mathbb{C}^d)$ it holds that
\begin{equation}
    \mathrm{Tr}[|\psi\rangle\langle \psi| P_{n,n}(|\psi\rangle\langle \psi|)]\geq 1-\frac{n^2}{d}.
\end{equation}
\end{lemma}
\begin{proof}
By Theorem~\ref{theorem:explicitdescriptionofirreps}
\begin{align}
\begin{split}
     \mathrm{Tr}[|\psi\rangle\langle\psi | P_{W_{n,n}}(|\psi\rangle\langle \psi|)]&= \mathrm{Tr}[|\psi\rangle\langle \psi| P_{L_n}(|\psi\rangle\langle \psi|)]- \mathrm{Tr}[|\psi\rangle\langle \psi| P_{L_{n-1}}(|\psi\rangle\langle \psi|)]\\
     &=1-\max_{\substack{M\in L_{n-1}\\ \|M\|_F=1}} |\mathrm{Tr}[M|\psi\rangle\langle \psi|]|^2.
     \end{split}
\end{align}
We will upper bound $|\langle \psi |M|\psi\rangle|^2$ for all 
Frobenius-normalized $M\in L_{n-1}$.
By Theorem~\ref{theorem:explicitdescriptionofirreps} we have $M=\Pi_{b,n}\mathbbm{1}_d\otimes \tilde{M}\Pi_{b,n}/N_{\tilde{M}}$, for some $\tilde{M}\in L(\mathrm{Sym}^{n-1}(\mathbb{C}^d))$, $\|\tilde{M}\|_F=1$ and $N_{\tilde{M}}=\|\Pi_{b,n}\mathbbm{1}_d\otimes \tilde{M}\Pi_{b,n}\|_F$ is a normalization factor. We can compute
\begin{align}
\begin{split}
|\langle \psi|M|\psi\rangle|^2&=\frac{|\langle \psi|\mathbbm{1}_d\otimes \tilde{M}|\psi\rangle|^2}{N_{\tilde{M}}^2}%\\
=\frac{\mathrm{Tr}\left[\mathrm{Tr}_{1}[|\psi\rangle\langle\psi|]\tilde{M}\right]}{N_{\tilde{M}}^2}\\
&\leq \frac{\|\mathrm{Tr}_{1}[|\psi\rangle\langle\psi|]\|_F\|\tilde{M}\|_F}{N_{\tilde{M}}^2}%\\
\leq \frac{\|\mathrm{Tr}_{1}[|\psi\rangle\langle\psi|]\|_1}{N_{\tilde{M}}^2}%\\
\leq \frac{1}{N_{\tilde{M}}^2},
\end{split}
\end{align}
using Cauchy-Schwarz and norm ordering $\|\cdot\|_F\leq \|\cdot\|_1$.
Further, we can compute
\begin{align}
\begin{split}\label{eq:normalizationcomputation}
    N^2_{\tilde{M}}&=\|\Pi_{b,n}\mathbbm{1}_d\otimes \tilde{M}\Pi_{b,n}\|^2_F%\\
    =\mathrm{Tr}[\Pi_{b,n}\mathbbm{1}_d\otimes \tilde{M}\Pi_{b,n}\mathbbm{1}_d\otimes \tilde{M}^{\dagger}]\\
    &=\mathrm{Tr}\left[\left(\frac{n-1}{n^2} \mathbb{F}_{1,2}\mathbbm{1}_d\otimes \tilde{M}+\frac{n-1}{n^2}\mathbbm{1}_d\otimes \tilde{M} \mathbb{F}_{1,2}+\frac{1}{n^2}\mathbbm{1}_d\otimes \tilde{M}+\frac{(n-1)^2}{n^2}\mathbb{F}_{1,2}\mathbbm{1}_d\otimes \tilde{M} \mathbb{F}_{1,2}\right)\mathbbm{1}_d\otimes \tilde{M}^{\dagger}\right]\\
    &=\frac{2n-2}{n^2}\mathrm{Tr}\left[\tilde{M}\tilde{M}^{\dagger}\right]+\frac{d}{n^2}\mathrm{Tr}\left[\tilde{M}\tilde{M}^{\dagger}\right]+\frac{(n-1)^2}{n^2}\mathrm{Tr}\left[\mathrm{Tr}_{n-1}[\tilde{M}]\mathrm{Tr}_{n-1}[\tilde{M}]^{\dagger}\right]%\\
    \geq \frac{d}{n^2},
    \end{split}
\end{align}
where the prefactors in the second line can be determined by counting permutations: The probability of a permutation to have the first copy of $\mathbb{C}^d$ as a fixed point is $1/n$, so the probability of the permutation left of $\mathbbm{1}_d\otimes \tilde{M}$ having the first copy of $\mathbb{C}^d$ not as a fixed point but the permutation to right does is $(1-1/n)(1/n)=(n-1)/n^2$.
But every permutation acting only on the copies $2,\ldots,n$ can be absorbed by $\tilde{M}$ or $\tilde{M}^{\dagger}$. 
Moreover, for any permutation $\pi$ that does not have the first copy of $\mathbb{C}^d$ as a fixed point, we can find permutations $\sigma_1,\sigma_2$ on the copies $2,\ldots,n$ such that $\sigma_1 \pi\sigma_2=\mathbb{F}_{1,2}$.
This leads to the first term $(n-1)/n^2 \mathbb{F}_{1,2} \mathbbm{1}_d\otimes \tilde{M}$.
The other terms follow similarly by counting permutations.
Eq.~\eqref{eq:normalizationcomputation} then completes the proof.
\end{proof}

The lemma implies large overlaps of a pure quantum state with the largest of the irreducible representations  $W_{n,n}$ in the ``collision free'' regime $d \gg n^2$.

\subsection{Bosonic, particle-number preserving interacting dynamics}
In \cite{oszmaniec_universal_2017} the authors showed that adding a single interacting (i.e non-quadratic) Hamiltonian to the set of non-interacting (i.e quadratic) Hamiltonians gives (for $d>2$ and $n\neq6$) a generating set that encompasses all unitaries on the bosonic subspace (for $d=2$ a further condition has to be satisfied). This tells us that if, in addition to the band-limited quadratic Hamiltonians we considered before, we add another non-quadratic Hamiltonian (for example a quartic one), we can generate all particle-number preserving interactions on $n$ bosons in most settings. 
The group generated by Hamiltonians including interaction terms acts as $\operatorname{SU}(d_F(n))$ on $L(\Sym^n(\mathbb C^d))$ (the space containing states of fixed particle number) via conjugation.  So for the group of (particle-number preserving) interacting unitaries we can set $\omega_n(U) = U \otimes \bar U$ representing $U \in \operatorname{SU}(d_F(n))$.  
An element in the group is specified by choosing an element of $\operatorname{SU}(d_F(n))$ for all $n$. 
E.g.~using the rules for Young diagrams as before, we find the following lemma: 
\begin{lemma}[Reduction of $\omega_n$ for interacting unitaries] \label{lem:interreps}
    The space $L(\Sym^n(\mathbb C^d))$ decomposes under the $\operatorname{SU}(d_F)$ action $U \otimes \bar U$ into a trivial representation carried by the identity on $\Sym^n(\mathbb C^d)$ and its complement $\sigma_n$ carried by $W_n \coloneqq L(\Sym^n(\mathbb C^d)) \setminus \langle \operatorname{Id}_{\Sym^n(\mathbb CC)} \rangle$.  We have $\dim W_n = d_F(n)^2 - 1 = {n + d - 1 \choose n}^2 - 1$.
\end{lemma}
One finds (e.g.~using the rules for Young diagrams as before) that the action of the unitary group by conjugation decomposes into a trivial irrep carried by the identity on $\Sym^n(\mathbb C^d)$ and its complement $W_n \coloneqq L(\Sym^n(\mathbb C^d)) \setminus \langle \operatorname{Id}_{\Sym^n(\mathbb C^d)} \rangle$.

It directly follows that the projector onto $W_n$ are given by 
\begin{align}
    \label{eq:projector-interactive}
    P_{n} = \operatorname{Id}_{L(\Sym^d(\mathbb C^d))} - \frac1{d_F} | \operatorname{Id}_{\Sym^d(\mathbb C^d)} ) ( \operatorname{Id}_{\Sym^d(\mathbb C^d)} | =  \sum_{j,k =1}^{d_F}\left[\ket{j,k'}\bra{j,k'} - \frac1{d_F} \ket{j,j'}\bra{k,k'} \right]\,
\end{align}
where $\ket{k, j'}$ denotes the basis for $\Sym^n(\mathbb C^d) \otimes \Sym^n(\mathbb C^d)$ as before.

\subsection{Fermionic, particle-number preserving dynamics}\label{app:fermionic}
The discussion of the fermionic case is analogous to the bosonic one. 
The Hilbert space of $n$ fermions in $d$ modes is given by the anti-symmetric tensor product $\wedge^n(\mathbb C^d)$.
In the absence of interactions, quadratic fermionic Hamlitonians also generate a group isomorphic to $\operatorname{SU}(d)$ but now acting by conjugation with anti-symmetric tensor powers $\wedge^n \otimes \overline{\wedge^n}$ on $L(\wedge^n(\mathbb C^d)$. 
This representation decomposes into  irreducible representations as follows. 

\begin{lemma}[Decomposition $\sigma^f_n$ for fermionic non-interacting transformations]
 The space $L(\wedge^n(\mathbb C^d))$ decomposes under the $\operatorname{SU}(d)$-action $\wedge^n \otimes \overline{\wedge^n}$ (i.e.\ the action of $\omega_n$ representing the group of particle-number preserving non-interacting unitaries on $n$ fermions in $d$ modes) into 
	into $\min\{d-n,n\}+1$ irreducible representations $\sigma^f_{n,l}$ carried by $W^f_{n,l}$ with $l \in \{0, \ldots, \min\{d-n,n\}\}$ of dimension $\dim W^f_{n,l} = {d\choose l}^2-{d\choose l-1}^2$.  
\end{lemma}	

\begin{proof}
	The Young diagram for the irreducible representation $\wedge^n$ is 
	{\scriptsize\begin{equation*}
	\left.\begin{ytableau}
	\\
	\\
	\vdots \\
	\\
	\\
	\end{ytableau}~\right\}~n.
	\end{equation*}}
	With the same rules for conjugates and tensor products of irreducible representations \cite{georgi_lie_1999}, we obtain for $\wedge^n \otimes \overline{\wedge^n}$:
	{\scriptsize\begin{equation}
	\left.\begin{ytableau}
	\\
	\\
	\vdots \\
	\\
	\\
	\end{ytableau}\right\}~n\quad \bigotimes \quad \left.
	\begin{ytableau}
	\\
	\vdots \\
	\\
	\end{ytableau}\right\}~d-n\quad =\quad 
	\left.\begin{ytableau}
	\\
	\\
	\vdots \\
	\\
	\\
	\end{ytableau}\right\} d\quad \bigoplus \quad
	\left.\begin{ytableau}
	~&\\
	~&\none\\
	\vdots&\none \\
	~&\none\\
	~&\none
	\end{ytableau}\right\}~d-1
	\quad \bigoplus\quad\dots \quad\bigoplus
	\left.\begin{ytableau}
	~&\\
	~&\vdots\\
	\vdots& \\
	~&\none\\
	~&\none
	\end{ytableau}\right\}~n.
	\end{equation}}
	This corresponds to $d-n+1$ many mutually non-isomorphic irreducible representations in the case that $n\geq d-n$. 
	If $n\leq d-n$, we can perform the same calculation with the tensor factors on the left side swapped. 
	This yields $n+1$ irreducible representations.
	In summary, we always have $\min\{d-n,n\}+1$ many irreducible representations.
	Again, one irrep always is the trivial representation.
    In the following, $l$ takes the values in $\{0,\ldots,n\}$ or $\{0,\ldots,d-n\}$ depending on whether $d-n\geq n$ or $d-n\leq n$.
	We denote the $l$-th irrep appearing in the decomposition above by $\sigma^f_{n,l}$ acting on the space $W^f_{n,l}$.
	We find the following formula for the dimensions of the corresponding irreducible representations:
	\begin{equation}
	\dim W^f_{n,l}:=\frac{d-2l+1}{d-l+1}{d\choose l}{d+1\choose l}={d\choose l}^2-{d\choose l-1}^2.
	\end{equation}
\end{proof}

Just as in the bosonic case, we can understand the irreducible representations  as the orthogonal complements or ``onion rings'' between embeddings of the linear maps acting on $r\leq n$ particles:

\begin{theorem}
    Let $P_{\wedge^n}$ be the projector onto $\wedge^n(\mathbb C^d)$.
    Consider the space of linear maps acting on $r<n$ particles defined as 
    \begin{equation}
     L^f_r\coloneqq \{P_{\wedge^n} (\mathbbm{1}_{d}^{\otimes n-r}\otimes M)P_{\wedge^n} , M\in L(\wedge^r(\mathbb{C}^d))\}.
    \end{equation}
    Then, $L^f_r$ is isomorphic as a subrepresentation to the space $L(\wedge^{r}(\mathbb{C}^d))$ equipped with the action of $\wedge^{r}$ by conjugation.
Moreover, for $d-n\geq n$ we find the projector onto $W^f_{n,l}$ to be 
\begin{equation}
    P^f_{n,l}\coloneqq P_{L^f_l}-P_{L^f_{l-1}}.
\end{equation}
For $n\geq d-n$, we find the reversed embedding: $L_r^f\subseteq L_{r-1}^f$ and therefore
\begin{equation}
    P_{n,l}=P_{L_{l-1}^f}-P_{L_{l}^f}.
\end{equation}
\end{theorem}
\begin{proof}
    The proof follows along the same lines as the proof of Theorem~\ref{theorem:explicitdescriptionofirreps}.
    The only difference is that we need to find a new example to show that $L^f_r$ properly contains $L^f_{r-1}$ in the case $d-n\geq n$.
    In this regime, we can choose the matrix 
\begin{equation}\label{eq:fermionicexample}
        P_{\wedge^n} \mathbbm{1}^{\otimes n-r}_d\otimes |1,2,\ldots,r\rangle\langle r+1,\ldots, 2r|P_{\wedge^n}=\sum_{j_1,\ldots,j_{n-r}}^m P_{\wedge^n}\bigotimes_{i=1}^{n-r} |j_i\rangle\otimes|1,2,\ldots,r\rangle \langle j_i|\otimes \langle r+1,\ldots, 2r|  P_{\wedge^n} ,
    \end{equation}
    for $r\leq n\leq d/2$ and where $|1,2,\ldots,r\rangle\in (\mathbb{C}^d)^{\otimes r}$ labels a tensor-product basis element. 
    We find that the decomposition of any matrix in $L_{r-1}$ into states $|x\rangle\langle y|$, with $|x\rangle,|y\rangle$ Fock states, contains only those $|x\rangle$ and $|y\rangle$ such that at most $r-1$ modes are occupied in $|x\rangle$ that are not also occupied in $|y\rangle$.
    Therefore, $L_{r-1}$ does not contain $L_r$ and we can repeat the argument from the proof of Theorem~\ref{theorem:explicitdescriptionofirreps}.
    This example breaks down for $m-n<n$ as the matrix in Eq.~\eqref{eq:fermionicexample} becomes zero.
We can assume that $n\leq d$ and reverse the roles of $n$ and $d-n$ for $d-n<n$.
\end{proof}

\section{Analog benchmarking via filtering}\label{app:filterRB}
Our randomized analog benchmarking protocol can be seen as an instance of the abstract concept of \emph{filtered randomized benchmarking (RB)} \cite{HelsenEtAl:2020:GeneralFramework,heinrich2023randomized,HelsenNewEfficientRB} that, e.g., also encompasses linear cross-entropy benchmarking~\cite{boixo_characterizing_2016} or matchgate~\cite{MatchgateBenchmarking} benchmarking.
As for any RB protocol~\cite{EmeAliZyc05, KnillBenchmarking,LevLopEme07, DanCleEme09,HelsenEtAl:2020:GeneralFramework}, the general approach is to apply random sequences of operations in an experiment and fit simple exponential decays to the obtained data.  
The decay parameters are then used as a figure of merit for the quality of the implementation of the employed operations. 
Under further assumptions, the decay parameters can be related to effective depolarizing strength and average fidelities of the noise channels, see e.g.\ \cite{KlieschRoth:2020:Tutorial}.  
Then, the protocols yield a theoretically interpretable characterization of the noise going beyond the mere benchmarking of the device's functioning. 

The central concept of filtered RB is to construct a scalar estimator using `filter functions' that isolate specific decays associated to individual irreps of the group generated by the gate-set. 
Therefore, filtered RB is particularly attractive and often a necessity for designing RB protocols for `smaller' sub-groups of the unitary group, where the action as a quantum channel decomposes into more than one non-trivial irrep. 
Without filtering, the resulting RB signal will be a linear combination of multiple decays and performing a fit becomes challenging.

Another crucial feature of filtered RB is that one does not need to apply an inversion gate at the end of the random sequence, as is custom in standard RB. 
Instead, the `inversion is done in the post-processing'. For this reason, filtered RB protocols generalize naturally to gate sequences that are drawn from a measure other than the uniform measure of a group.  
The general protocol and general theoretical guarantees for such filtered RB protocols that use random circuits of gate sets have been developed in Ref.~\cite{heinrich2023randomized}. 
The blue-print laid out in Ref.~\cite{heinrich2023randomized} facilitates taking a bottom-up approach to designing RB protocols: 
Start from native device operations, work out the corresponding filtered RB scheme and understand the parameter regime in which the protocol is guaranteed to function. 

As we argue in this work, the flexibility of filtered RB schemes is essential for benchmarking non-universal analog quantum devices.  
In the following sections, we develop our analog benchmarking protocol for interacting and non-interacting particle number preserving dynamics following this blue-print, derive explicit expressions for the implementation and work out concrete theoretical guarantees. 
We then evaluate these guarantees for specific ensembles of bosonic Hamiltonians that might arise in a device implementation. 

In this section, we provide a detailed derivation of the randomized analog benchmarking protocol.
The result of this section is summarized in Table~\ref{table:summary-protocol-ingredients} in the main text.

We begin by restating the device setting and the randomized analog benchmarking (RAB) protocol.
We consider a dynamic analog quantum simulator that is (e.g.) designed to apply quenches, where the system evolves under a quartic Hamiltonian of the form Eq.~\eqref{eq:quartic hamiltonian} for a constant time. 

These native operations might further be restricted by the device design and physics to a 
certain interaction graph, e.g.\ nearest-neighbour hopping and only local interactions as 
\begin{align}
	H(h,V) = \sum_{\substack{i,j=1 \\ i=j \text{ or } i=j\pm 1}}^d h_{i,j} a_i^\dagger a_j + \sum_{i=1}^d V_{i,i,i,i} a_i^\dagger a_i^\dagger a_i a_i + h.c\,. 
\end{align}
For the purpose of benchmarking, we define a random ensemble of quenches according to the device's specifications.
For example, in terms of simple uniform probability distributions of the hopping and interaction parameters 
\begin{align}
	h_{i,i} \sim [-a,a], & & h_{i,i\pm 1}=[-b,b], & & V_i = [-c,c],
    \label{eq:band-diagonal-hamiltonian}
\end{align}
where $x \sim [a, b]$ indicates that $x$ is drawn uniformly at random from the interval $[a,b] \subset \mathbb R$. 
For benchmarking non-interacting dynamics we set  $V_{i}=0$.
This defines a measure $\nu$ on the Hermitian matrices or fixing a time $\Delta t$ on the unitaries acting on the entire Fock-space via $U = \mathrm e^{i \Delta t H}$ with $H \sim \nu$.  
Slightly overloading notation, we refer to the measures on the Hamiltonian parameters, the Hermitian matrices, and the unitaries simply as a random ensemble of quenches $\nu$. 
It will be clear from the context which one we are referring to. 
This measure will be our working example throughout this work.  
It is however important to stress that it is straightforward to adapt the measure to the experimental platform at hand. 
The ideal dynamics that we are benchmarking is particle-number preserving.  This motivates to benchmark the device for each fixed number $n$ of excitations separately.
In particular, we assume that the device is capable of performing measurements in the Fock basis targeting a certain fixed number $n$ of excitations. 
In principle, it is also possible to probe irreps with different $n$ simultaneously using different measurements.

\renewcommand{\algorithmicensure}{\textbf{Return:}}
\begin{algorithm}[H]
\caption{\label{alg:rabb} Random analog benchmarking (RAB)}
\begin{algorithmic}[1]
\Require Measure $\nu$, set of sequence lengths $\mathcal M$ ($m \geq m_\text{\rm ths}$ for all $m \in \mathcal M$), 
    initial state $|\psi_0\rangle\in \Sym^n(\mathbb C^d)$, number of iterations $K$, $N$
    \Statex{\textbf{Quantum experiment:}}
    \For{$m \in \mathcal{M}$} $K$ iterations \textbf{of}
    \State Prepare initial quantum state $\rho_0 = |\psi_0\rangle\langle\psi_0|$
    \State Draw random Hamiltonians $H_\gamma\sim \nu$, $\gamma \in \{1, \ldots, m\}$.
    \For{$N$ iterations}
    \State Implement time evolution  $U = e^{-i\Delta tH_m}\cdots e^{-i\Delta t H_1}$
    \State Measure in Fock basis $\{E_x\}_{x=1}^{d_F}$
    \EndFor
    \State Estimate relative frequencies of measurement outcomes $\hat p=(\hat p_1,\hat p_2,\ldots,\hat p_{d_F})$.
    \State Store $(m, U, \hat p)$. 
    \EndFor
    \Statex{\textbf{Classical postprocessing for irrep $\lambda$:}}
    \State For every data point $(m, U, \hat p)$ compute multi-shot $\lambda$-filter function 
    \begin{equation}
        \hat{F}_{\lambda}(m, U) = \frac{1}{n_{\lambda}} \sum\limits_{x=1}^{d_F} f_{\lambda}(x, U) \hat p_x%p_{x|U} .
    \end{equation}
    \State Compute sequence average $\overline{F}_{\lambda}(m) = \frac{1}{K}\sum\limits_{k=1}^K  \hat{F}_{\lambda}(m, U_k)$ over all $U_k$ with the same $m$ for every $m\in\mathcal M$.
    \State Fit an exponential decay
     $$ \overline{F}_{\lambda}(m) =_{\text{fit}} A_{\lambda} z_{\lambda}^m. $$
    \Ensure Decay parameter $z_\lambda$.
\end{algorithmic}
\end{algorithm}

The random measure $\nu$ is the only non-trivial ingredient for the data acquisition stage of the RAB protocol.  
We state the complete protocol consisting of an experiment on the quantum device for the data acquisition and a post-processing stage using a classical computer as Protocol~\ref{alg:rabb}\,.
The experimental data is post-processed weighting the relative outcome frequencies with a filter function $f_{\lambda}$ depending on which irreps of the underlying group are investigated. The post-processing requires a specification of the filter function $f_\lambda$, the normalization factor $n_\lambda$ and the minimal admissible sequence length $m_\text{\text{ths}}$. 

\paragraph*{Irrep labels $\lambda$.} 
We consider two random ensembles generating different groups:
(i) Interacting system with hopping and interaction terms (ii) Non-interacting system with only hopping (quadratic) terms. 
As derived in App.~\ref{sec:representation-theory-for-particle-number-preserving-dynamics},
the irreducible representations of (the group generated by) interacting Hamiltonians are labelled with their particle number $\lambda = n$ with $n\in{}\mathbb N^+$.
The irreps of (the group generated by) non-interacting Hamiltonians are labelled with a tuple as $\lambda = (n,l)$ with $l\in \{1, \ldots, n\}$ and $n\in\mathbb N^+$.

\paragraph*{Filter functions $f_\lambda$.}\label{par:filter-functions} 
For each irrep, we define the filter function that associates a scalar value to a unitary and a measurement outcome. 
The filter function is calculated as the action of the ideal sequence within a specific irrep. 
Denote the projectors of the $n$-excitation Fock basis by $\{E_x = |x \rangle \langle x |\}_{x=1}^{d_{F}}$ and the dephasing channel associated to the measurement by
\begin{align}
    M = \sum_{x=1}^{d_F}|E_x)(E_x| \, .
    \label{eq:dephasing-channel}
\end{align}
Ref.~\cite{heinrich2023randomized} defines the following general expression for the filter function of a multiplicity-free irrep  $\lambda$ and an ideal initial state $\rho_0$:
\begin{align} \label{eq:filter-function-appendix}
		f_\lambda(x,U) \coloneqq s_\lambda^{-1}(P_\lambda(\rho_0), U^\dagger E_x U)\,, 
\end{align}
where $s_\lambda =  \tr(P_\lambda M)/\tr(P_\lambda)$ and $P_\lambda$ the projector onto the irrep $\lambda$.

Using the expressions for the dimension of $\lambda$ and $P_\lambda$ derived in App.~\ref{sec:representation-theory-for-particle-number-preserving-dynamics}, we can calculate the filter functions for our settings:

(i) For the ensemble of random quenches with interactions, by Lemma~\ref{lem:interreps} the prefactor becomes
$s_n = (d_F^2 - 1)^{-1}(d_F - 1) = (d_F+1)^{-1}$
and using Eq.~\eqref{eq:projector-interactive} the filter function is given by 
\begin{align}
\begin{split}
    f_n(x, U)
    &=(d_{F} +1) \langle \psi_0 \otimes \bar \psi_0 | U\otimes \bar U  P_n | x\otimes x\rangle, \\
    &= (d_{F} +1)\left( |\langle x |U|\psi_0\rangle|^2 -\frac{1}{d_{F}}\sum_{j,k=1}^{d_{F}}\langle \psi_0\otimes\bar\psi_0|U\otimes \bar U|j\otimes j \rangle\langle k\otimes k |x\otimes x\rangle \right), \\
    &=(d_{F} + 1) \left( |\langle x|U|\psi_0 \rangle|^2 - \frac{1}{d_F}\sum_{k=1}^{d_{F}}|\langle k|U |\psi_0 \rangle|^2 \right), \\
    &=(d_F + 1) \left( |\langle x|U|\psi_0 \rangle|^2 - \frac{1}{d_{F}} \tr_n(\rho_0) \right)\,,
    \label{eq:flambda-interacting}
    \end{split}
\end{align}

where we denote by $\tr_n(\rho_0) = \sum_{j =1}^{d_F} \langle j | \rho_0 | j \rangle$ the trace restricted to the irrep $\lambda = n$ and use its invariance under conjugation by $U$. 
Note the (non-accidental) striking resemblance of the filter function with a linear cross-entropy estimator.

(ii) For benchmarking non-interacting quenches, the filter function onto  $\lambda=(n,l)$ can be expressed using Eq.~\eqref{eq:projector-passive-gaussian} for the projector in terms of Clebsch-Gordon coefficients.
Then the filter function Eq.~\eqref{eq:filter-function-appendix} can be simplified as
\begin{align}
\begin{split}
    f_{n,l}(x, U)
    &= \frac{1}{s_{n,l}}\langle\psi_0\otimes\psi_0|U\otimes\bar U P_{n,l} |x\otimes x\rangle\,, \\
    &=  \frac{1}{s_{n,l}}\sum_{a''_\lambda} \sum_{a,a',b, b'} C^{a''_\lambda}_{a,a'} C^{a''_\lambda}_{b,b'}
       \langle \psi_0 \otimes \bar \psi_0 | U\otimes \bar U |a \otimes a'\rangle\langle b \otimes b'| x\otimes x\rangle\,, \\
    &=  \frac{1}{s_{n,l}}
       \langle \psi_0 \otimes \bar \psi_0 | U\otimes \bar U\sum_{a''_\lambda} C^{a''_\lambda}_{x,\bar x}\sum_{a,a'}C^{a''_\lambda}_{a,a'} |a \otimes a'\rangle\,, \\
    &=\frac{1}{s_{n,l}} (\rho_0 |\omega_n(U)  | \alpha_\lambda(x)),
       \label{eq:flambda-noninteracting}
       \end{split}
\end{align}
where we defined $ \sum_{a''_\lambda} C^{a''_\lambda}_{x,\bar x}\sum_{a,a'}C^{a''_\lambda}_{a,a'} |a \otimes a'\rangle = |\alpha_\lambda(x))$.

The normalization factor $s_\lambda$ for $\lambda=(n,l)$  is calculated to be 
\begin{align}
    s_{n,l} = \frac{\tr( P_\lambda M)}{\tr P_\lambda}
    =\tfrac1{\tr P_\lambda}\!\!\!\!\!\sum_{a_\lambda'',a,a',b,b',x}\!\!C^{a_\lambda''}_{a, a'}C^{a_\lambda''}_{g, g'} \langle x,\bar x)|a,a'\rangle\langle b,b'|x,\bar x)\rangle
    = \frac{d-1}{(d+2l-1){d+l-2\choose l}^2} \sum_{a_\lambda'',x} \left(C^{a_\lambda''}_{x, \bar x}\right)^2,
    \label{eq:slambda-noninteracting}
\end{align}
where we inserted the dimension of $W_{n,l}$ as derived in Lemma~\ref{lem:noninterirreps}.

\paragraph*{Normalization factor $n_\lambda$.}
We choose the normalization $n_\lambda$ such that the filtered RAB signal $\bar F_\lambda(m)$ is $1$ in expectation when the quantum device is working perfectly and the unitaries $U$ were drawn from the uniform (Haar) measure $\mu$ over the respective groups.

The Born probability for observing outcome $x$ in the state $U\rho_0U^\dagger= U |\psi_0\rangle\langle\psi_0|U^\dagger$ is $p_x = (E_x| \omega_n(U) |\rho_0)$ since $E_x, \rho_0 \in L(\Sym^n(\mathbb C^d))$. Thus, in the absence of noise for each $
U$ we have $\mathbb E[\hat p_x]= p_x$.
Thus,  
\begin{align}
	n_{\lambda} 
	= \mathop{\mathbb E \;}_{U\sim\mu} \left[\sum_{x=1}^{d_F} f_{\lambda}(x,U) p_x\right]
	= \tfrac1{s_{\lambda}} \mathop{\mathbb E \;}_{U\sim\mu} \sum_{x=1}^{d_F}(\rho_0|P_{\lambda}\omega^\dagger_n(U)|E_x)(E_x|\omega_n(U)|\rho_0)
	= \tfrac1{s_{\lambda}}\mathop{\mathbb E \;}_{U\sim\mu} (\rho_0|\mathcal \omega^\dagger_n(U) P_\lambda M \omega_n(U)|\rho_0).
 \label{eq:normalization-per-irrep}
\end{align}
Since $X \mapsto \mathop{\mathbb E}_{U\sim\mu} \omega^\dagger_n(U) X \omega_n(U)$ is the orthogonal projector onto the commutant of the image of $\omega_n$, we have according to Schur's lemma 
$\mathop{\mathbb E}_{U\sim\mu} \omega^\dagger_n(U) X \omega_n(U) = \sum_\lambda ( P_\lambda, X)P_\lambda/\tr(P_\lambda)$, where the sum is over the irreps in the decomposition of $\omega_n$.
Therefore, we find that
\begin{align}
    n_{\lambda}
    = \frac{1}{s_{\lambda}}(\rho_0|\sum_{\lambda'} \frac{\tr(P_{\lambda'}P_{\lambda}M)}{\tr(P_{\lambda'})} P_{\lambda'}|\rho_0)
    =  \frac{ \tr(P_{\lambda} M)} {s_\lambda \tr(P_\lambda)} (\rho_0| P_\lambda |\rho_0)
	= (\rho_0|P_{\lambda}|\rho_0).
    \label{eq:normalization}
\end{align}
For a pure initial state $|\psi_0\rangle$ 
we find for $\lambda=n$ that $n_n=(1 - 1/d_F)$ and  for $\lambda=(n,l)$ that
\begin{align*}
    n_{n,l} 
    =  \sum_{a_\lambda'',a,a',b,b'} C^{a_\lambda''}_{a, a'}C^{a_\lambda''}_{b, b'} \langle \psi_0,\bar\psi_0|a,a'\rangle\langle b,b'|\psi_0,\bar\psi_0\rangle
    = \sum_{a_\lambda'', x}\left(C_{x,\bar x}^{a_\lambda''}\right)^2 |\langle x|\psi_0 \rangle|^2\,.
\end{align*}

\paragraph*{Minimal sequence length $m_{\text{ths}}$.} One parameter defining the data acquisition step is the minimal admissible sequence length. 
This parameter depends on the `speed' at which $\nu$ converges to a sufficiently uniform measure as experienced within the representation $\lambda$ under consideration. 
We provide a careful analysis of our measure $\nu$ in Section~\ref{secapp:convergence-analysis}.

\section{Fitting model and signal guarantees}\label{Ap.:Fitting_model_and_signal_guarantees}

We now justify the underlying assumption of the protocol that the RB signal is well-described by a single real-valued exponential decay. 
The expected RB signal is defined as the expected value of the filter function $f_\lambda$ taken over sequences of random unitary quenches drawn from a measure $\nu$ and the relative measurement outcome frequencies $\hat p$,
\begin{equation}\label{eq:expected signal}
    F_{\lambda}(m) = \mathbb{E} [\hat{F}_{\lambda}(m)] = n_\lambda^{-1}\mathop{\mathbb E \;}_{U\sim\nu} \left[\sum_{x=1}^{d_F} f_{\lambda}(x,U) p_x\right]\,,
\end{equation}
where we here used that $\hat p$ in expectation is the Born-probability $p$.

We model the actual implementation of a unitary quench by a map $\phi$. 
When attempting to implement the unitary quench $U$, the actual process on the device will be 
a completely positive, trace-preserving map $\phi(U)$ acting on $L(\Sym^n(\mathbb C^d))$. 
Note that assuming the existence of such a map $\phi$ implicitly assumes that time-dependent, time-correlated or non-Markovian noise effects can be neglected---a common assumption for randomized benchmarking protocols.  
Furthermore, we assume that the implemented quenches are particle-number preserving.  
We will discuss at the end of the section the implications of relaxing this assumption in detail.

Denoting the noisy initial state by $\tilde\rho_0 \in L(\Sym^n(\mathbb C^d))$  and the noisy measurement effects by $\tilde E_x \in L(\Sym^n(\mathbb C^d))$ for $x=1,\ldots, d_f$, the Born probability of measuring outcome $x$ after applying $m$ many unitary quenches is $p_x = (\tilde E_x| \phi(U_m)\cdots\phi(U_1)|\tilde\rho_0)$.  As $\phi(U)$ acts on $L(\Sym^n(\mathbb C^d))$, we can assume w.l.o.g. that $\tilde \rho_0$ and $\tilde E_x$  are restricted to $n$th Fock sector. We can simply regard $\tilde \rho_0$ and $\tilde E_x$ as the projection of a general initial state and measurement effects acting on the entire Fock space.

In the following, we find a more concise expression for the expectation value of the filter functions.
Recall the definition of the filter function Eq.~\eqref{equ:single-filter-function-expression-main-text} from the main text and refer to Table~\ref{table:summary-protocol-ingredients} and Section~\ref{app:filterRB} of the appendix for the constants and notation used throughout the derivation.
By inserting the Born probability of the noisy implementation and writing out the filter function the executed expected value reads as 
\begin{align} 
   \mathbb{E} [\hat{F}_{\lambda}(m)] =  n_\lambda^{-1}\mathbb E\left[f_\lambda(x, U\right] 
    = n_\lambda^{-1} 
     \mathop{\mathbb E \;}_{U_1,\ldots, U_m\sim\nu} \sum_{x=1}^{d_F}
     s_\lambda^{-1}(\rho_0| P_\lambda
    \omega_n(U^\dagger) |E_x)(\tilde E_x| \phi(U_m)\cdots\phi(U_1) 
    |\tilde\rho_0).
    \label{eq:expectation-value-of-filter-function}
\end{align}

We denote the noisy implementation of the dephasing channel in the measurement basis $\tilde M = \sum_{x=1}^{d_F}|E_x)(\tilde E_x|$ and use the fact that $P_\lambda \omega_n(\,\cdot\,) = \sigma_\lambda(\,\cdot\,)$ (as introduced in App.~\ref{sec:representation-theory-for-particle-number-preserving-dynamics}), where we always understand all maps as being embedded as maps on $L(\Sym^n(\mathbb C^d))$.

Further, denoting the sequence of unitary quenches as $U = U_m\cdots U_{1}$ we arrive at the following expression:

\begin{align}
    \mathbb{E} [\hat{F}_{\lambda}(m)] =  \frac{1}{n_\lambda s_\lambda}(\rho_0|
     \mathop{\mathbb E \;}_{U_1,\ldots, U_m\sim\nu}
    \sigma_\lambda(U_1)^\dagger\cdots\sigma_\lambda(U_m)^\dagger \tilde M \phi(U_m)\cdots\phi(U_1) 
    |\tilde\rho_0).
\end{align}

The expression $\tilde T_\lambda: X \mapsto \mathbb E_{U\sim\nu}\sigma_\lambda(U)^\dagger X \phi(U)$ is the `noisy' channel twirl over the measure $\nu$ restricted to the irrep $\lambda$ written as $\tilde T_\lambda(\,\cdot\,)$.  
This twirl map is applied $m$ many times on the noisy dephasing channel $\tilde M$:
\begin{align}\label{eq:expected signal T}
     \mathbb{E} [\hat{F}_{\lambda}(m)] = \frac{1}{n_\lambda s_\lambda}(\rho_0| \tilde T_\lambda^m(\tilde M)|\tilde\rho_0) =   \frac{1}{n_\lambda s_\lambda}(Q| \tilde T_\lambda^m |\tilde M)\,,
\end{align}
where $Q = \rho_0\otimes\tilde\rho_0$. Note that this expression holds for gate-dependent noise and an arbitrary measure $\nu$.

It is instructive to look at some common special cases first.  When assuming gate-independent noise and $\nu$ being the Haar measure $\mu$ on the group generated by the quenches, we find the following result:

\begin{lemma}[Guarantee for Haar measures and gate-independent noise]\label{lem:decay_haar_random}
    Let $\tilde{\rho}_0$ and $\tilde M$ be the SPAM implementation. %$\tilde{E}_x$ represent the initial state and measurement effects including (SPAM) noise, where $x = 1, \dots , d_F$ and $d_F$ is the dimension of the Fock space. 
    Set $\tilde M =\sum_x |E_x)(\tilde E_x|$. 
    Assume gate-independent noise such that $\phi(U)= \Lambda_L \omega_n(U) \Lambda_R$ with $\Lambda_L, \Lambda_R$ endomorphisms on $L(\Sym^n(\mathbb C^d))$ and
     $\nu = \mu$ the Haar measure on $\operatorname{SU}(d)$ (or $\operatorname{SU}(d_F)$ when including interactions).
    The expected RB signal Eq.~\eqref{eq:expected signal} for $\lambda = (n,l)$ (or $\lambda = n$, respectively,) 
    decays exponentially in the sequence length $m$ as
\begin{align}
F_\lambda(m) = A_\lambda z_\lambda^m\, 
    && \text{where} &&
    A_\lambda=\frac{(\rho_0 |P_\lambda\Lambda_L|\tilde\rho_0)
        (\Lambda_RP_\lambda|\tilde M)}{n_\lambda s_\lambda\tr(P_\lambda\Lambda_L\Lambda_R)}
        && \text{and} &&
        z_\lambda=\frac{1}{\tr(P_\lambda)}
        \tr(P_\lambda\Lambda_L\Lambda_R)\,.
\end{align}

\end{lemma}
\begin{proof}
    We begin by restating the expectation value expression 
    $\mathbb E\left[\hat F_\lambda(m)\right] = n_\lambda^{-1} s_\lambda^{-1}(Q| \tilde T_\lambda^m |\tilde M)$.
     For gate-independent noise and $\nu = \mu$, the noisy channel twirl restricted to the irrep $\lambda$ can be evaluated as
     % \IR{$\mathbb{E}$}
    \begin{align}
        \tilde T_\lambda
        = \mathbb{E}_{U\sim\mu}\sigma_\lambda(U)^\dagger\otimes\phi(U)^T
        = (\mathrm{id}_\lambda\otimes\Lambda_L^T){\mathbb E}_{U\sim\mu}\sigma_\lambda(U)^\dagger\otimes\omega^T_n(U)(\mathrm{id}_\lambda\otimes\Lambda_R^T) 
        = (\mathrm{id}_\lambda\otimes\Lambda_L^T) \frac{|P_\lambda)(P_\lambda|}{\tr(P_\lambda)} (\mathrm{id}_\lambda\otimes\Lambda_R^T),
    \end{align}
   where in the last step we used that by invariance of the Haar measure for $T_\lambda: X \mapsto \mathbb E_{U\sim \mu} \sigma^\dagger_\lambda(U) X \omega_n(U)$ we have
    $[\sigma_\lambda(V), T_\lambda] = 0$ for all $V \in \operatorname{SU(d)}$ for $\lambda = (n, l)$ (or $V \in \operatorname{SU}(d_F)$ for $\lambda = n$) and Schur's lemma. 

    Inserting this expression into Eq.~\eqref{eq:expected signal T} yields
    \begin{align}
        \mathbb E\left[\hat F_\lambda(m)\right]
        &= \frac{1}{n_\lambda s_\lambda}(Q| 
        \left((\mathrm{id}_\lambda\otimes\Lambda_L^T) \frac{|P_\lambda)(P_\lambda|}{\tr(P_\lambda)} (\mathrm{id}_\lambda\otimes\Lambda_R^T)\right)^m 
        |\tilde M), \label{eq:lemma6-derivation-insert-Tisprojecor}\\
        &= \frac{1}{n_\lambda s_\lambda\tr(P_\lambda)}
        (Q| (\mathrm{id}_\lambda\otimes\Lambda_L^T) |P_\lambda)
        \left((P_\lambda|(\mathrm{id}_\lambda\otimes\Lambda_R^T)(\mathrm{id}_\lambda\otimes\Lambda_L^T) |P_\lambda)\right)^{m-1}
        (P_\lambda|(\mathrm{id}_\lambda\otimes\Lambda_R^T)|\tilde M),\\
        &= \frac{1}{n_\lambda s_\lambda\tr(P_\lambda)}
        (Q|P_\lambda\Lambda_L)(\Lambda_RP_\lambda|\tilde M)
        \left(\frac{1}{\tr(P_\lambda)}(P_\lambda|P_\lambda\Lambda_L\Lambda_R)\right)^{m-1}, \\
        &=  \frac{1}{n_\lambda s_\lambda\tr(P_\lambda\Lambda_L\Lambda_R)}
        (\rho_0 |P_\lambda\Lambda_L|\tilde\rho_0)
        (\Lambda_RP_\lambda|\tilde M)
        \left(\frac{1}{\tr(P_\lambda)}
        \tr(P_\lambda\Lambda_L\Lambda_R)\right)^m.
    \end{align}
     Identifying $z_\lambda$ and $A_\lambda$ as in the lemma's statement completes the proof.
\end{proof}

The lemma is a special case of the more general statements derived in Ref.~\cite{HelsenEtAl:2020:GeneralFramework} for uniform randomized benchmarking with arbitrary groups. 

Having limited ranges of Hamiltonian parameters and restricted connectivity, an analog simulator will typically not be able to natively implement Haar random quenches.  
Adapting a more digital computing mindset one could attempt to compile Haar random quenches out of the hardware-native set of quenches.  
For example, a product of quenches with ${d \choose 2}$ non-interacting particle-number preserving Hamiltonians, randomly chosen such that they apply uniform quenches between neighbouring modes on a one-dimensional chain, is uniformly distributed on $\operatorname{SU}(d)$ \cite{DiscreteApproximation}.  This is a consequence of the universality of band-diagonal Hamiltonians following from Hurwitz' Lemma \cite{Hurwitz}. 

For benchmarking analog devices, a considerably more flexible approach is to directly construct benchmarking sequences using a measure $\nu$ on native quenches.   
Using a native measure $\nu$ instead of the Haar measure, the twirl $T_\lambda$ does not simplify to an outer product of the projectors onto the irrep $\lambda$.
This complicates the derivation above.

Before stating a general guarantee for gate-dependent noise and non-uniform $\nu$, we consider the gate-independent depolarising channel.
Let $\mathcal D_p: X \mapsto p X + (1- p) \tr[X] \operatorname{Id} / d_F$ denote the depolarizing channel on $L(\Sym^d(\mathbb C^d))$, with $p$ as depolarising strength, applied gate $\phi(U) = \mathcal D_p\circ\omega_n(U)$.
To simplify the argument, we further assume perfect SPAM and arrive at the following statement.

\begin{lemma}[RAB signal with depolarising noise]\label{lem:rab-signal-under-depolarising-noise}
    Assume global depolarising noise with strength $p$,
    then the RAB signal is of the form
    $p^m(1 + \alpha(m))$ with $\alpha(m)$ independent of $p$ and $\alpha(m) \to 0$ for $m \to \infty$.
\end{lemma}
\begin{proof}

Under our noise assumption, we have
\begin{align}\label{eq:depolnoisecalc1}
   \sum_{x=1}^{d_F} f_\lambda(x,U) p_x
    = s_\lambda^{-1}(\rho_0| P_\lambda \omega_n(U^\dagger) \tilde M \mathcal D_p\omega_n(U_m)\cdots\mathcal D_p\omega_n(U_1) |\tilde\rho_0).
\end{align}
Note that the projector obeys $P_\lambda(\mathbb I) = 0$ for $\lambda = (n, l)$ and $\lambda = n$ and that $\mathbb I/d_F$ is a fixed point of $M$ and $\omega_n(U)$ for all $U$.  
Thus, inserting the definition of the depolarizing channel in Eq.~\eqref{eq:depolnoisecalc1} all terms involving a maximally mixed state $\mathbb I/d_F$ vanish. 
Evaluating the expression gives only a single summand for depolarizing noise
\begin{align}
     s_\lambda\sum_{x=1}^{d_F} f_\lambda(x,U) p_x
    = p^m (\rho_0| P_\lambda \omega_n(U^\dagger) M\omega_n(U)|\rho_0)\,.
\end{align}
For a (gapped) measure $\nu$, we can write its channel twirl $T_\lambda: X \mapsto \mathbb E_{U \sim \nu} \sigma_\lambda(U) X \omega_n(U)$ as $T_\lambda= |P_\lambda)(P_\lambda|/\tr(P_\lambda) + \Lambda_\lambda$ with $(P_\lambda|\Lambda_\lambda = \Lambda_\lambda|P_\lambda) = 0$ and $\|\Lambda\|_\infty \leq \Delta_\lambda$ for some $0 \leq \Delta_\lambda < 1$. Here the first term corresponds to the twirl's expression for the Haar measure that we used e.g.\ in determining the normalization factor $n_\lambda$. 
Hence, 
\begin{equation}
    F_\lambda(m) = \tfrac1{n_\lambda s_\lambda} \mathbb E_{U\sim \nu} \sum_x f_\lambda(x, U)p_x = \frac{p^m}{n_\lambda s_\lambda} (Q|T_\lambda^m|M) = p^m + \frac{p^m}{n_\lambda s_\lambda} (Q | \Lambda_\lambda^m| M) = p^m ( 1 + \alpha(m))\,,
\end{equation}
where the last step defines the scalar function $\alpha$.  
Since $|\alpha(m)| \leq \text{const.}\ \Delta^m_\lambda$ with some (dimension-depend) constant in $m$ and $\Delta_\lambda < 1$, $F_\lambda(m) \approx p^m$ for sufficiently large $m \geq m_\text{\rm ths}$.  
Furthermore, for depolarizing noise this subdominant term $\alpha(m)$
is independent of the noise strength.
\end{proof}

For `simple' noise $m_\text{\rm ths}$ can, thus, be determined, e.g., by numerical simulation of the protocol without any noise.   
This will inform our approach in Sec.~\ref{subsec:simulating-warm-up-phase-with-RAB-protocol}.

\paragraph*{General guarantee}
We established that the expected RAB signal for Haar randomly drawn quenches follows an exponential decay, Lemma \ref{lem:decay_haar_random}.
Moreover, for a simple noise model, this approximation holds 
for a non-uniform measure $\nu$ provided that $m \geq m_\text{\rm ths}$. 
This behaviour also holds in the presence of gate-dependent noise. 
The signal of sufficiently long sequences of non-uniformly sampled quenches converges to the signal for the uniform measure on the generated group.  
For a general filtered randomized benchmarking protocol, Ref.~\cite{heinrich2023randomized} showed that sampling non-uniformly from a compact group leads to an exponential decay after a minimal sequence length  $m_\text{\rm ths}$.
This minimal length depends on the measure and a perturbative noise assumption. 
This justifies fitting an exponential decay for a sequence of $m \geq m_\text{\rm ths}$.
Using the results of Ref.~\cite{heinrich2023randomized}, we can derive the following general guarantee for the RAB protocol. 
To this end, recall that we denote by $T_\lambda : X \mapsto \mathbb E_{U \sim \nu} \sigma_\lambda(U) X \omega_n(U)$ and $\tilde T_\lambda : X \mapsto \mathbb E_{U \sim \nu} \sigma_\lambda(U) X \phi(U)$. 
\begin{theorem}[Signal form guarantee for non-interacting quenches]\label{theorem:signal-guarantee-nonint}
Choose $n \in \mathbb N$, $l \in \{1, \ldots n\}$. Let $\nu$ be a measure on $\operatorname{SU}(d)$ (supported on quenches with non-interacting Hamiltonians) that is an approximate unitary $2n$-design with spectral gap $\Delta^{(2n)}_\nu$.  
Assume that there exist $\delta \in (0, \Delta^{(2n)}_\nu / 5)$ such that the implementation fulfils 
\begin{equation}\label{eq:guarantee assumption}
    \|T_{n,l} - \tilde T_{n,l}\|_\infty \leq \delta\, .
\end{equation}
Then, the expected signal of the RAB protocol is of the form
\begin{equation}\label{eq:decay_curve_random_sequences}
    F_{n,l}(m) = A_{n,l} z^m_{n,l} + \alpha(m)\,
\end{equation}
with real $z_{\lambda} \in [1 - 2\delta, 1]$ independent of SPAM errors  
and $|\alpha(m)| \leq \check\alpha$ for 
\begin{equation}\label{eq:guarantee sequence length}
    m \geq m_\text{\rm ths} = \frac2{\Delta^{(2n)}_\nu} \left( \frac23 \log \dim(W_{n,l}) + \frac12 \log \frac1{n_{n,l}} + \log \frac1{\check\alpha} + 1.8\right)\,, 
\end{equation}
with $\dim(W_{n,l})$ the dimension of the irrep and $n_{n,l}$ the normalization factor Eq.~\eqref{eq:normalization}.
\end{theorem}

\begin{proof}
The theorem follows from adapting Ref.~\cite[Theorem 8]{heinrich2023randomized} to the case at hand and controlling the spectral gap of the twirl operator for all irreps $(n,l)$ with $l \in \{1, \ldots, n\}$ simultaneously by $\Delta^{(2n)}_\nu$.  

To this end, let $T_{\lambda, \mu} : X \mapsto \mathbb E_{U \sim \mu} \sigma_\lambda(U) X \omega_n(U)$ and $\Delta_\lambda = \| T_{n,l} - T_{(n,l), \mu}\|_\infty$ the spectral gap of $T_\lambda$ to the Haar average.  All irreducible representations $\sigma_{n,l}$ are self-dual, so $\Delta_\lambda$ is also the spectral gap of
$
\mathbb{E}\sigma_{n,l}(U)\otimes \omega_n(U)$.
Now $U \mapsto \sigma_{n,l}(U)\otimes  \tilde{\omega_n}(U)$ 
is a subrepresentation of $\omega_n^{\otimes 2}$, which in turn is a subrepresentation of $U \mapsto U^{\otimes 2n}\otimes \overline{U
}^{\otimes 2n}$.
Therefore, $T_{n,l}$
block diagonalizes into irreducible representations that are also contained in $U \mapsto U^{\otimes 2n}\otimes \overline{U}^{\otimes 2n}$ and the gap $\Delta^{(2n)}_{\nu} \leq \Delta_\lambda $ lower bounds the gap of $T_{n,l}$. 

Thus, Eq.~\eqref{eq:guarantee assumption} ensures that the corresponding condition Ref.~\cite[Theorem 30]{heinrich2023randomized} is met for all $l$ and the convergence of the sub-dominant contribution to the RAB signal is also controlled with $\Delta^{(2n)}_{\nu}$.  The signal form, a scalar decay, then follows from observing that $\sigma_{n,l}$ appears in $\omega_n$ with multiplicity one (Theorem~\ref{theorem:explicitdescriptionofirreps}).  Further, $|\alpha(m)| \leq c_{n,l} ( 1 - \Delta^{(2n)}_\nu + 2 \delta)^m$.  
The constant is bounded by $c_{n,l} \leq e^{1.8} \sqrt{\dim(W_{n,l})} / (s_{n,l} \sqrt{n_{n,l}})$, where we have used that our definition of the filter function includes an additional factor of $n_{n,l}^{-1}$ compared to the definition in Ref.~\cite{heinrich2023randomized}. Using $s_{n,l}\geq 1/\dim W_{n,l}$, we have $c_{n,l} \leq e^{1.8} \dim(W_{n,l})^{3/2} / \sqrt{n_{n,l}}$. 
Thus, $m \geq \frac2{\Delta^{(2n)}_\nu} (\log c_{n,l} + \log 1/{\check \alpha})$ implies that $|\alpha(m)| \leq \check\alpha$.  Expanding $\log c_{n,l} \leq 3/2 \log \dim(W_{n,l}) - 1/2 \log n_{n,l} + 1.8$ yields the theorem's statement. 
\end{proof}

Let us unpack the statement of this theorem. 
The term $\alpha(m)$ in Eq.~ \eqref{eq:decay_curve_random_sequences} needs to be suppressed sufficiently in order to be able to fit an exponential decay confidently. The rate at which this suppression happens is governed by $\delta$ and $\Delta^{(2n)}_\nu$.  We observe rapid decays for small $\delta$ and large spectral gap $\Delta^{(2n)}_\nu \approx 1$. 
In the case where $\nu$ is the Haar measure (or at least a $2n$-design) the moment operator is a projector and indeed $\Delta_{\lambda}=1$. In this case, Eq.~\eqref{eq:guarantee assumption} can be understood as a condition on the strength of tolerated noise.
This ensures that the $\tilde T_{n,l}$ is still well approximated by a unit rank operator.  
In particular, we will observe a scalar decay as soon as the unit rank approximation dominates the RAB signal. 
When $\nu$ is a non-uniform measure, such a unit rank approximation is still accurate after a sufficiently long sequence length and if the noise is benign enough. 
More precisely, we require the spectral gap to stay open under noise.
In a perturbative regime where $\delta \leq \Delta_{\nu}^{(2n)} / 5$, matrix perturbation theory  ensures  that the spectral gap is still larger than $\Delta_\nu^{(2n)} - 2\delta$ \cite{heinrich2023randomized}. 
Thus, for a sequence length that scales inversely with the spectral gap we continue to observe a single scalar exponential decay. Eq.~\eqref{eq:guarantee sequence length} gives a precise condition for controlling the subdominant term $\alpha$ up to additive precision $\check \alpha$. In particular, we find a lower bound on $m_{\text{ths}}$ that scales logarithmically in the irrep's dimensions.  
Thus, we find a sufficient bound on the sequence length $m_{\text{ths}} \gtrsim (\Delta_\nu^{(2n)})^{-1} l \log d$. 

The setting with interacting quenches follows analogously. 
Here we need to ensure that the measure induces an approximate unitary $2$-design on $\operatorname{SU}(d_F)$.  
Since $\sigma_n$ also appears in $\omega_n$ without multiplicities, we find a scalar exponential decay after a sequence length scaling inversely with the design's spectral gap and logarithmically in $\dim W_n = d_F^2 -1$.

It is important to notice that observing single exponential decays in our RAB protocol relies on the fact that the irreps appear in $\omega_n$ without multiplicities.   
From the decomposition Lemma~\ref{lem:noninterirreps} we know that the irrep $(n, l)$ in the decomposition of $\omega_n$ has an isomorphic irrep $(n', l)$ in the decomposition of $\omega_{n'}$ for all $n' \geq l$.  
Thus, if $\phi$ is not particle-number preserving it is in principle possible that it intertwines different equivalent irrep spaces. 
The RAB signal then takes the form of a matrix-exponential decay $F_{n,l}(m) = \tr[A I^m]$ with a square matrix $I$ depending on $\phi$ and a matrix $A$ depending on the SPAM errors. 
The dimension of $I$ is the number of irreps that are intertwined by $\phi$  \cite{heinrich2023randomized}.

Furthermore, since $I$ might be only diagonalizable over $\mathbb C$, $F_{n,l}(m)$ can be non-monotonous and oscillating.
If the RAB signal (for $m \geq m_{\text{ths}}$) is found to consist of multiple decays or oscillations in praxis, this can, hence, be understood as a signature of violating the assumption of particle number preserving noise.

\section{Relations for spectral gaps}\label{secapp:convergence-analysis}

The guarantee of the RAB protocol, Theorem \ref{theorem:signal-guarantee-nonint}, for arbitrary gate-dependent noise only makes use of the spectral gap $\Delta^{(2n)}_\nu$ (for non-interacting quenches, $\Delta^{(2)}_\nu$ for interacting quenches) to specify a sufficient minimal sequence length $m_{\text{ths}}$ for a non-uniform measure $\nu$. 
Given a measure $\nu$ it is possible to numerically obtain a bound for $\Delta^{(2n)}_\nu$.  This can for example be achieved by empirically estimating the corresponding frame potential Eq.~\eqref{eq: frame_potential}.  
We in this section develop the required results for relating these quantities.

\begin{lemma}[Lower bound to spectral gap]\label{lemma:lower-bound-spectral-gap}
    Let $F_{\nu}^{t}(m)$ be the frame potential for a sequence of random quenches of length $m$. The spectral gap $\Delta_\nu^{(t)}$ is lower bounded as 
    \begin{equation}
        \Delta_\nu^{(t)} \geq 1 - \sqrt[2m]{\mathcal{F}_\nu^{(t)}(m) - t!} := S_{\nu}^{(t)}.
    \label{eq:lower-bound-spectral-gap}
    \end{equation}
\end{lemma}
\begin{proof} 

The derivation boils down to proving the two inequalities
\begin{align}
    i) \;\;&   \left\|\left(M_{\nu}^{(t)}\right)^m - M_\mu^{(t)}\right\|_{\mathrm{F}}^2
    \geq (1-\Delta_\nu^{(t)})^{2m},\\
    ii)\;\;&  \left\|\left(M_{\nu}^{(t)}\right)^m - M_\mu^{(t)}\right\|_{\mathrm{F}}^2
    = \mathcal{F}_\nu^{(t)}(m) - \mathcal{F}_\mu^{(t)},
\end{align}
where $\mu$ is the Haar measure.

Starting with the left-hand side of $(i)$, using the left and right invariance of the Haar measure and the fact that the moment operator $M_\nu^{(t)}$ is hermitian due to the symmetric measure $\nu$ it follows that
\begin{align}
    \left\|\left(M_{\nu}^{(t)}\right)^m - M_\mu^{(t)}\right\|_{\mathrm{F}}^2
    = \mathrm{tr}\left[
        \left(M_{\nu}^{(t)}\right)^{2m}
        - \left(M_{\nu}^{(t)}\right)^m M_\mu^{(t)}
        - M_\mu^{(t)}\left(M_{\nu}^{(t)}\right)^k
        + \left(M_\mu^{(t)}\right)^2
    \right].
\end{align}
Using the left and right invariance of the Haar measure  $M_\nu^{(t)}M_\mu^{(t)}=M_\mu^{(t)}M_\nu^{(t)} = M_\mu^{(t)}$ and an eigenvalue decomposition of both the moment operators $M_{\nu}^{(t)} = U\Sigma_\nu^{(t)}U^\dagger$ leads to
\begin{align}
    \left\|M_{\nu^{*m}}^{(t)} - M_\mu^{(t)}\right\|_{\mathrm{F}}^2
    = \mathrm{tr}\left[\left(M_{\nu}^{(t)}\right)^{2m} - M_{\mu}^{(t)} \right]
    = \mathrm{tr}\left[ \left(\Sigma_\nu^{(t)}\right)^{2m} - \Sigma_\mu^{(t)} \right]
    = \sum_{i=1}^{d_F^{2t}}\left(\left(\lambda_{\nu,i}^{(t)}\right)^{2m} - \lambda_{\mu,i}^{(t)}\right).
    \label{eq:sum-eigvalues-momentoperators}
\end{align}
Choose the order of each moment operator eigenvalues such that 
$\Sigma_\mu^{(t)} = \text{diag}(\underbrace{1, 1, \ldots, 1}_{t!}, 0, 0, \ldots, 0)$
and
$\Sigma_\nu^{(t)}=\text{diag}(\underbrace{1, 1, \ldots, 1}_{t!}, \lambda_{\nu,2}^{(t)}, \lambda_{\nu,3}^{(t)}, \ldots, \lambda_{\nu,n}^{(t)})$ with $1\geq\lambda_{\nu,2}^{(t)}\geq\ldots\geq\lambda_{\nu,n}^{(t)} \geq 0$ where $n=d_F^{2t} - t!$.
The spectral gap is defined as the gap between the eigenvalues which are $1$ and the biggest eigenvalue $\lambda_2$.
With that Eq.~\eqref{eq:sum-eigvalues-momentoperators} simplifies to 
\begin{align}
    \sum_{i=2}^{d_F^{2t} - t!}\left(\lambda_{\nu,i}^{(t)}\right)^{2m}
    \geq \left(\lambda_{\nu,2}^{(t)}\right)^{2m}
    = (1-\Delta_\nu^{(t)})^{2m},
\end{align}
showing $(ii)$. 
We denote by $\nu^{\ast m}$ the $m$-fold convolution of the measure.  The product of $m$ unitaries i.i.d.\ drawn from $\nu$ is distributed according to $\nu^{\ast m}$.
Using the invariance of the Haar measure $(ii)$ follows from 
\begin{align}
     \left\|\left(M_\nu^{(t)}\right)^m - M_\mu^{(t)}\right\|_{\mathrm{F}}^2
    &= \mathrm{tr}\left[\left(M_\nu^{(t)}\right)^m\left(\left(M_\nu^{(t}\right)^\dagger\right)^m - M_\mu^{(t)}\right]\,, \\
    \nonumber 
    &=\mathop{\mathbb E \;}_{\substack{V_1,\ldots, V_m\sim\nu \\ U_1,\ldots, U_m\sim\nu}} \left[
    \tr\left( \left(U_1^{\otimes t}\cdots U_m^{\otimes t} \left(V_1^{\otimes t}\right)^\dagger\cdots \left(V_m^{\otimes t}\right)^\dagger \right)
    \otimes 
    \left(\overline U_1^{\otimes t}\cdots \overline U_m^{\otimes t} \left(V_1^{\otimes t}\right)^T\cdots \left(V_m^{\otimes t}\right)^T \right)\right]\right) - \mathcal F_\mu^{(t)}\,,\\
    \nonumber 
    &= \mathop{\mathbb E \;}_{U, V\sim\nu^{*m}}
    \tr\left(UV^\dagger\right)^t
    \overline{\tr\left(UV^\dagger\right)}^t - F_\mu^{(t)}\,,\\
    \nonumber 
    &=  \mathcal{F}_{\nu^{*m}}^{(t)} - \mathcal{F}_\mu^{(t)}\,.
    \nonumber 
\end{align}
\end{proof}

\subsection{Sample complexity}\label{app:subsec:sample-complexity}
We next turn to the sample complexity. We focus only on the number of random sequences that have to be applied at each sequence length to get an accurate estimate of the sequence-averaged RAB signal.
RAB gives you an error parameter for each irreducible representation.
We have shown in App.~\ref{app: overlap}, the largest overlap an initial pure state has is with the projector of the largest irreducible representation.
This also indicates that the error parameter of the largest irreducible representation is the most accurate information one can gain from this protocol.
Therefore, we find the tightest bound on the sample complexity for the largest irreducible representation.

\begin{theorem}[Sample complexity for non-interacting quenches]
    Let $\overline{F}_{\lambda}(m)$ be the unbiased estimator, using $K$ samples, of the sequence average $F_{\lambda}(m)$ defined in Eq.~\eqref{eq:expected signal}. Choose $\epsilon >0$, $\delta \in (0,1)$ and suppose that
    \begin{equation}
        K \geq \frac{\operatorname{dim}(W_{n,n})^2}{2\epsilon^2 d_F^2(1 - n^2/d)^4} \log(2/\delta)\,.
    \end{equation}
    Then in the collision-free regime $d \gg  n^2$, it holds that $\left | \overline F_{n,n}(m)- F_{n,n}(m)\right | \leq \epsilon$ with probability $1-\delta$.
\end{theorem}

\begin{proof}
    Using Hoeffding's inequality for a random variable $\hat X \in [a, b]$ and the mean estimate $\overline X = 1/K \sum_{k=1}^{K}\hat X_k$ we have 
    \begin{equation}
		\Pr\left[\left|\overline X-\mathbb{E}[\hat X]\right|\geq \epsilon\right]\leq 2 \exp{\frac{-2K\epsilon^2}{(b-a)^2}} \geq \delta,
	\end{equation}
 which means that in order to achieve a sample estimate $\overline X$ which is at least $\epsilon$ close to the expected value with confidence $\delta$ we require 
	\begin{equation}
		K\geq \frac{ (b-a)^2}{2\epsilon ^2}\log(2/\delta) .
	\end{equation}
Recall that the RB signal from Eq.~\eqref{eq:Fhat-lambda} is
\begin{equation}
\hat{F}_{\lambda}(m) = n_{\lambda}^{-1} \sum\limits_{x=1}^{d_F} f_{\lambda}(x, U) \hat p_{x} \;\in [a,b], \text{ with } f_{\lambda}(x, U) = s_{\lambda}^{-1} (\rho_0|P_{\lambda}\omega_n(U)^\dagger|E_x).
\end{equation}

In the single shot regime set $\hat p_{x|U} = \delta_{x,x'}$ with $x' \sim p_{x|U}$. Therefore, we obtain an upper bound $b$ with
\begin{align}
    \hat F_{\lambda}(m)
    = n_{\lambda}^{-1} \sum\limits_{x=1}^{d_F} f_{\lambda}(x, U) \delta_{x,x'}
    =  n_{\lambda}^{-1}f_{\lambda}(x', U) \leq \frac{1}{n_{\lambda}s_{\lambda}}.
\end{align}

For the collision free regime $d \gg n^2$ and in the biggest irreducible representation $\lambda = (n,n)$ we use  Lemma~\ref{lower_bound_overlap} and find $(\rho_0|P_{n,n}|\rho_0) \geq 1 - n^2/d$ and $\mathrm{tr}(P_{n,n}M) \geq d_F(1 - n^2/d)$.
Therefore, $b$ is lower bounded as
\begin{align}
    \frac{1}{n_{n,n}s_{n,n}}
    = \frac{\operatorname{dim}(W_{n,n})}{(\rho_0| P_{n,n} |\rho_0)\mathrm{tr}(P_{n,n}M)}
    \leq \frac{\operatorname{dim}(W_{n,n})}{d_F(1-n^2/d)^2}
    := b.
\end{align}

With that and the assumption that $a = -b$ the sample complexity $K$ can be written as 
\begin{align}
    K \geq \frac{2b^2 }{\epsilon^2} \log(2/\delta) = \frac{2\operatorname{dim}(W_{n,n})^2}{\epsilon^2 d_F^2(1 - n^2/d)^4} \log(2/\delta).
\end{align}
\end{proof}

Note that the dimension of the Fock space is $d_F = \binom{n+d-1}{n}$ and for the sake of clarity we will carry the number of particles $n$ in the notation as $d_F^n = \binom{n+d-1}{n}$.
Then the dimension of the carrier space reads as $\operatorname{dim}(W_{n,n}) = \left(d_F^n\right)^2 - \left(d_F^{n-1}\right)^2 = \binom{d+n-1}{n}^2 - \binom{d + n - 2}{n-1}^2$.

\begin{align}
     \frac{\operatorname{dim}(\lambda_n)^2}{d_F^2} = \frac{\left(d_F^{n}\right)^4 - 2 \left(d_F^{n}\right)^2\left(d_F^{n-1}\right)^2 + \left(d_F^{n-1}\right)^4}{\left(d_F^{n}\right)^2}\leq 2 \left(d_F^{n}\right)^2 \leq 2e^{2n}(n+d-1)^{2n} / n! = O(d^{2n}),
\end{align}
which gives us the following corollary.
\begin{corollary}[{Sample efficiency in collision-free regime}] \label{cor:collision-free-sample-complexity}
    In the collision-free regime, $d \gg n^2$, $F_{n,n}(m)$ is w.h.p.\ an $\epsilon$-accurate estimator with a number of samples polynomialy in $d$ and for $n \in O(1)$. 
\end{corollary}

\section{Numerical simulations and methods}\label{app:numerical-simulations}

Here, we give details on the numerical simulations presented in the main text and present further results. 
The numerical simulations have been performed in Python and the code is available as Ref.~\cite{Jadwiga_github}.
We numerically study: First, the frame potential to estimate the spectral gap, second, the warm-up phase of a random ensemble, and third the RAB protocol itself with different noise models.

We consider the time evolution under two random ensembles of Bose-Hubbard Hamiltonians with a one-dimensional hopping graph defined in Eq.~\eqref{eq:band-diagonal-hamiltonian} in the main text. 
\begin{enumerate}
    \item \emph{Uniform ensemble}: uniformly sampled on-site potentials $h_{i,i}\sim[-1,1]$, hopping strength $j=i\pm1$ are complex numbers where both the real as well as the imaginary part are drawn from the interval $[-1,1]$. The evolution time is $\Delta t=1$.
    The interaction is modelled as on-site interaction with the distribution  $V_{i}\sim[-1,1]$.
    \item \emph{Sycamore ensemble}: uniformly sampled diagonal on-site potentials $h_{i,i}\sim [-20\mathrm{MHz},20 \mathrm{MHz}]$, and fixed hopping entries $h_{i,i+1}=-20 \mathrm{MHz}$ and time evolution $\Delta t=25 \mathrm{ns}$.
    The on-site interaction terms are always set to $V_{i}~=~-5\mathrm{MHz}$. The parameters are taken from the Google Sycamore quantum chip \cite{hangleiter2021precise} with a different interaction strength.
\end{enumerate}

A table for the used Clebsch-Gordan coefficients for the RAB protocol simulations is given in Tab.~\ref{tab:cgc_4d_2d}.

The error bars and absolute errors of each data point were computed using a non-parametric bootstrapping method.
To use the bootstrapping method, for each specified batch of data, the data was resampled with the same batch size $100$ times for the RAB protocols and $1000$ times for the frame potential estimations; this was done simultaneously for each data point. For any fitting process, the mean of each resampled data batch was calculated and fitted.
These fitting parameter estimates were then used to calculate an absolute error for each parameter.
Error bars for the means, the decay parameters, and the spectral gaps were determined at a confidence interval of $.95$.

\subsection{Numerical representation}
For the simulation, we directly work in the representation of the Hamiltonian dynamics in the particle-preserving subspace of bosonic systems with $d$ modes and $n$ particles where the Fock dimension is denoted as $d_F$. 
We use as the standard basis the set of Fock basis vectors $\{f_k\}_{k=1}^{d_F}$ with
\begin{align}
    |f\rangle = |n_1,n_2,\ldots,n_d \rangle
    = \left(b_1^\dagger\right)^{n_1}\cdots\left(b_d^\dagger\right)^{n_d}|0,0,\ldots,0\rangle, \text{ given } \sum_{k=1}^{d}n_k = n,
\end{align}
with the bosonic creation and annihilation operators 
\begin{align}
     & b_k^\dagger|f\rangle = \sqrt{n_k+1} |n_1,\ldots,n_k+1 ,\ldots,n_d, \rangle, &
     b_k|f\rangle = \sqrt{n_k} |n_1,\ldots,n_k-1 ,\ldots,n_d \rangle\,.
\end{align}
In this basis, the matrix representation of the Bose-Hubbard Hamiltonian from Eq.~\eqref{eq:quartic hamiltonian} is given as 
\begin{align}
    (H(h,V))_{k,l}  &= \langle f_k | H(h,V) | f_l\rangle.
\end{align}
For example, the matrix representation of the Hamiltonian of a $3$ mode system with $2$ particles with on-site interaction and nearest-neighbour hopping becomes
\begin{equation*}
H(h,V) = 
\left[\begin{matrix}2 V_{1} + 2 h_{1,1} & \sqrt{2} h_{1,2} & 0 & 0 & 0 & 0\\\sqrt{2} {h_{1,2}^*} & h_{1,1} + h_{2,2} & h_{2,3} & \sqrt{2} h_{1,2} & 0 & 0\\0 &  {h_{2,3}^*} & h_{1,1} + h_{3,3} & 0 & h_{1,2} & 0\\0 & \sqrt{2}  {h_{1,2}^*} & 0 & 2 V_{2} + 2 h_{2,2} & \sqrt{2} h_{2,3} & 0\\0 & 0 &  {h_{1,2}^*} & \sqrt{2}  {h_{2,3}^*} & h_{2,2} + h_{33} & \sqrt{2} h_{23}\\0 & 0 & 0 & 0 & \sqrt{2}  {h_{2,3}^*} & 2 V_{3} + 2 h_{33}\end{matrix}\right],
\end{equation*}
where $^*$ denotes complex conjugation.

\subsection{Numerical estimation of the frame  potential and evaluation of bounds on \texorpdfstring{$m_\text{\rm ths}$}{mths}}
\label{subsec:estimation-of-mths}

\begin{figure}
    \begin{floatrow}
     \ffigbox[0.7\textwidth]
       {\includegraphics[width=1\linewidth]{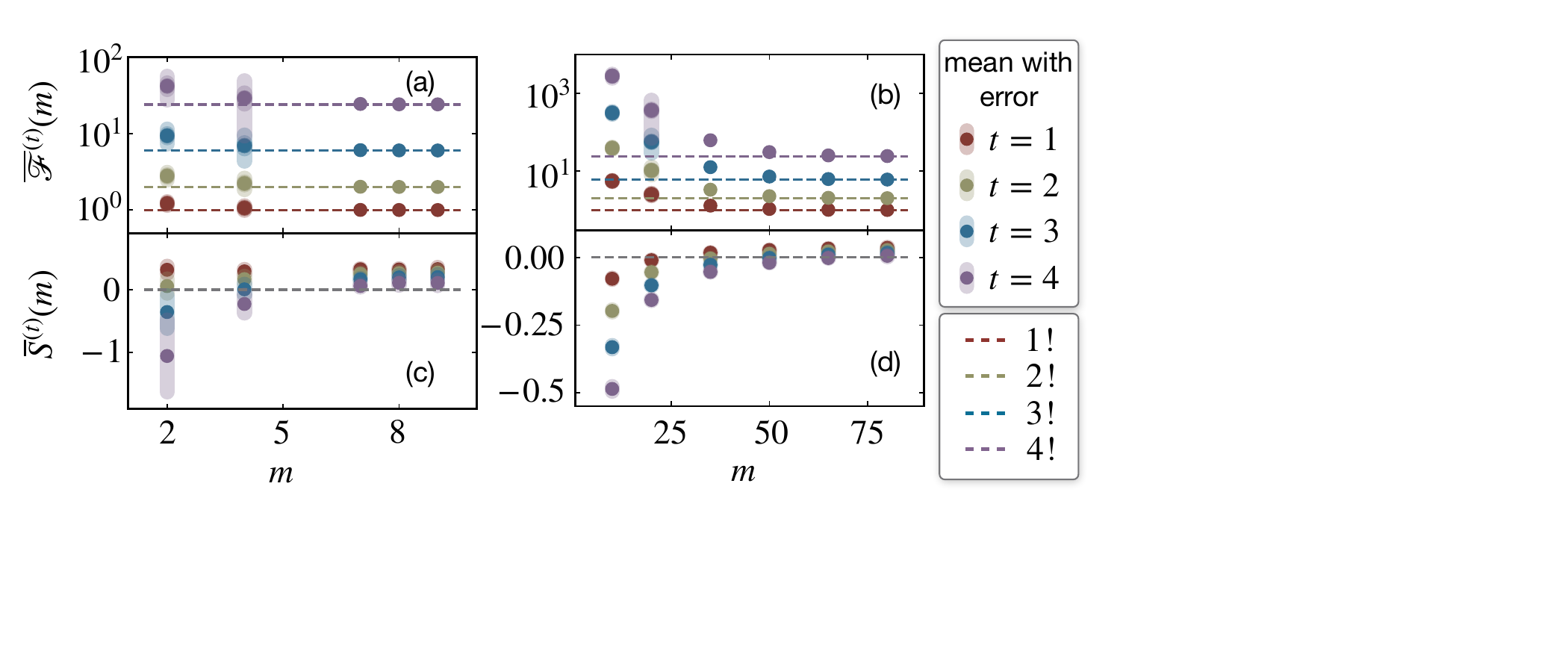}}{
       \caption{\emph{Estimating the frame potential and spectral gap.} The LHS displays results for the uniform setting, and the RHS results for the sycamore setting. In the top frame, the frame potential Eq.~\eqref{eq:statistical-estimate-frame-potential} as a function of $m$ is presented, these values are then converted into the spectral gap estimates plotted in the lower panel as a function of $m$, see Eq.~\eqref{eq:numerical-estimate-spectral-gap}.
       The dashed lines indicate the analytical frame potential values for Haar-random unitary ensembles.} 
       \label{fig:numerics-framepotential-spectralgap}
       }
     \capbtabbox[\Xhsize]
       {
       \setlength{\extrarowheight}{5pt}
       \begin{tabular}{r|cc} \toprule
   & uniform & sycamore \\ \hline
   $m=M$ & $9$ & $80$ \\
   $\overline{\mathcal{F}}^{(4)}(M)$ &  $24.13$  & $24.39$\\
   $\Delta\overline{\mathcal{F}}^{(4)}(M)$ &  $0.12$  & $0.21$\\
   $\hat\Delta_{\nu,\text{min}}^{(4)}$ & $ 0.11$ & $0.05$\\
   $\Delta\hat\Delta_{\nu,\text{min}}^{(4)}$ & $ 0.05$ & $0.03$\\
   $\hat m_\mathrm{thr}, \lambda=(2,1)$ & $156$ & $342$ \\
   $\hat m_\mathrm{thr}, \lambda=(2,2)$ & $178$ & $393$ \\ \bottomrule
  \end{tabular}
  }{
  \caption{
    \emph{Numerical values} of the frame potential Eq.~\eqref{eq:statistical-estimate-frame-potential}, spectral gap Eq.~\eqref{eq:numerical-estimate-spectral-gap} and analytical minimal sequence length Eq.~\eqref{eq:estimate-minimal-sequence-length} from Fig.~\ref{fig:numerics-framepotential-spectralgap}.
    A $\Delta$ indicates the absolute error.
    }
  }
  \label{table:numerics-framepotential-spectralgap}
    \end{floatrow}
   \end{figure}

The minimal admissable sequence length $m_\mathrm{ths}$ of the RAB protocol is controlled by the spectral gap of the chosen ensemble of random quenches in Theorem~\ref{theorem:signal-guarantee-nonint}. 
Particularly, the spectral gap of the $2n$-th moment of the ensemble, where $n$ denotes the particle number. 
In this section, we numerically estimate the frame potential for different random ensembles and determine a lower bound estimate on the spectral gap using Lemma \ref{lemma:lower-bound-spectral-gap}.  
This allows us to evaluate the analytical bound on $m_\mathrm{ths}$ of Theorem~\ref{theorem:signal-guarantee-nonint}.

The estimator of the frame potential for each quench depth $m$ and moment $t$ is calculated using
\begin{align}
    \overline{\mathcal{F}}^{(t)}(m) :=
    \overline{\mathcal{F}}_{\nu^*m}^{(t)} = \frac{1}{K}\sum_{k=1}^{K}\left|\mathrm{tr}\left[U_kV_k^\dagger\right]\right|^{2t} \text{ with } U_k, V_k \sim \nu^{*m},
    \label{eq:statistical-estimate-frame-potential}
\end{align}
where drawing $U\sim\nu^{*m}$ denotes the $m$-fold convolution of $\nu$, i.e.\ drawing $m$ many unitaries, where a single unitary quench is calculated as 
$U=U_1U_2\cdots U_m$ for $U_1,U_2,\ldots,U_m\sim\nu$.

Using the relation~\eqref{eq:lower-bound-spectral-gap} we then calculated the numerical estimate of each spectral gap as $\overline S_{\nu}^{(t)}(m) = 1 - \sqrt[2m]{\overline{\mathcal{F}}_\nu^{(t)}(m) - t!} $.

In Fig.~\ref{fig:numerics-framepotential-spectralgap} the estimates of the frame potentials up to the fourth moment and the estimates of the spectral gap are shown, both for the uniform and the sycamore setting for non-interacting systems.
For small quench depths $m$ the estimator $\overline S$ is negative and continuously grows until it is higher than zero.
Once the value rises above zero, including the bootstrapped error bars, the spectral gap will continue to increase for longer quench depths. However, sampling out the estimator for larger depth values is costly, for example, depth values of $m=9$ for the uniform setting require up to $10^7$ samples. 
From here, we can construct a lower bound on the minimal analytical spectral gap $\Delta_{\nu,\text{min}}^{(t)} $ using Eq.~\eqref{eq:lower-bound-spectral-gap} based on our numerical estimates, the estimated spectral gap $\hat\Delta_{\nu,\text{min}}^{(t)}$ and the maximum quench depth $M$ for which the estimate can be found numerically:

\begin{equation}\label{eq:numerical-estimate-spectral-gap}
    \Delta_{\nu,\text{min}}^{(t)} \geq \hat\Delta_{\nu,\text{min}}^{(t)} \coloneqq 1 - \sqrt[2M]{\overline{\mathcal{F}}_\nu^{(t)}(M) - t!}.
\end{equation}

Using this value, we then calculate the minimal sequence length from which the signal is guaranteed to describe a scalar exponential decay using the expression from Theorem~\ref{theorem:signal-guarantee-nonint}.

If we set the strength with which the noise disturbs the decay signal to $\check\alpha=10^{-2}$ we find the minimal sequence length to be 
\begin{equation}\label{eq:estimate-minimal-sequence-length}
    \hat m_\text{\rm ths} = \frac2{\hat\Delta^{(2n)}_\nu} \left( \frac23 \log \dim(W_{n,l}) + \frac12 \log \frac1{(\rho_0|P_\lambda|\rho_0)}  + 6.5\right)\,, 
\end{equation}
where $\dim(W_{n,l})$ is the dimension of the irrep and we have used Eq.~\eqref{eq:guarantee sequence length}.

A summary of the numerical values of the frame potentials and their associated spectral gap estimates along with the resulting minimal sequence length of the numerical simulations displayed in Fig.~\ref{fig:numerics-framepotential-spectralgap} can be found in Table~\ref{table:numerics-framepotential-spectralgap}.

\subsection{Simulating warm-up phase with the RAB protocol}
\label{subsec:simulating-warm-up-phase-with-RAB-protocol}

In the above subsection, we saw that the bounds of the minimal admissible sequence length from Theorem~\ref{theorem:signal-guarantee-nonint} is rather large. 
These would present challenges for implementing the RAB protocol in practice for sizeable noise.  The signal strength for long sequence of quenches is easily too weak to yield  much information from the quantum device. 
Numerical simulations, however, indicate that the derived worst-case guarantees might be not tight or overly pessimistic. 
Assuming a simple noise model, we derived a subdominant behaviour (deviation from a scalar decay) for the signal in Section~\ref{Ap.:Fitting_model_and_signal_guarantees} Lemma~\ref{lem:rab-signal-under-depolarising-noise} that is independent of the noise strength. 
In this subsection, we will use this observation  to present another way to determine the minimal admissible sequence length.

The idea is the following: After choosing the desired unitary ensemble, the RAB protocol is simulated with an absence of noise, revealing the convergence of the signal to $1$.
Once the absolute value of the difference of the signal to $1$ drops below the tolerated systematic error, this sequence length can be used as a heuristic minimal sequence length for the warm-up phase.

In Fig.~\ref{fig:numerics-warmup-phase} four different  settings are simulated using ideal quenches without noise or SPAM errors.
In the section above we chose $\check\alpha$, the tolerated systematic error, to be $10^{-2}$.  
We observe that the signal deviates from is asymptotically value by less $\check\alpha = 10^{-2}$ at considerably smaller sequence length that the bounds from Theorem~\ref{theorem:signal-guarantee-nonint}. 
We find minmial admissible sequence length of $16$, $8$, and $55$ for the interacting ensemble, non-interacting uniform, and Sycamore ensembles, respectively.

\begin{figure}
    \centering
    \includegraphics[width=1\linewidth]{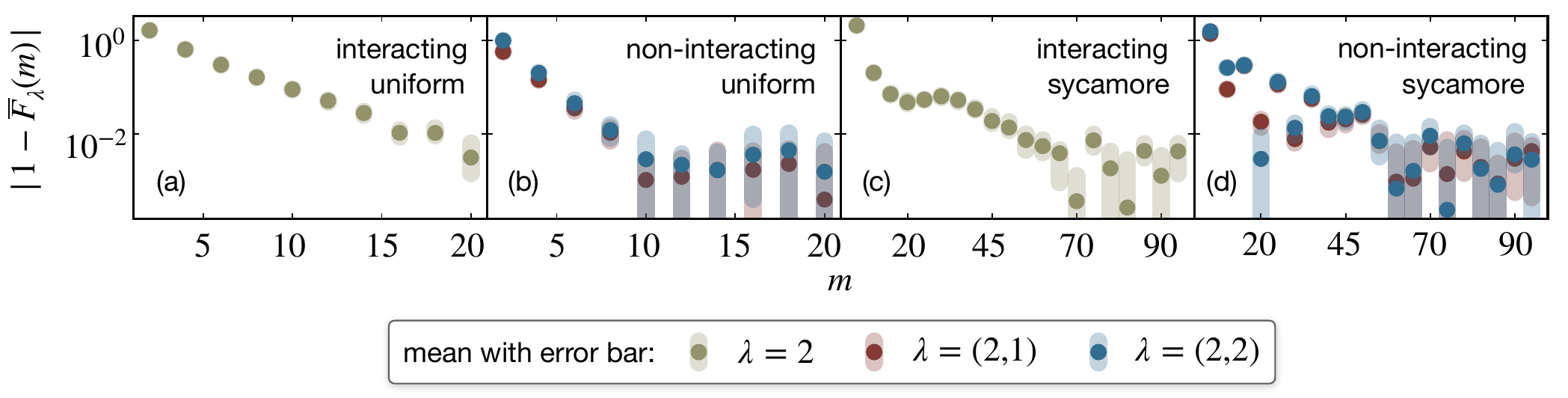}
    \caption{\emph{RAB simulations without noise} to determine the heuristic warm-up phase. The absolute value of the difference of the signal to $1$ is shown against the quench depth $m$. Once this drops under a certain value the corresponding depth can be used as a heuristic warm-up phase.
    Then the values for $m_\mathrm{thr}$ are (a) $16$, (b) $8$, (c) $55$, (d) $55$. They are used}
    \label{fig:numerics-warmup-phase}
\end{figure}

\subsection{Simulating the RAB protocol with different noise models}
\label{paragraph:simulating-bosonic-quantum-system}

\begin{figure}[ht!]
    \centering
    \includegraphics[width=\linewidth]{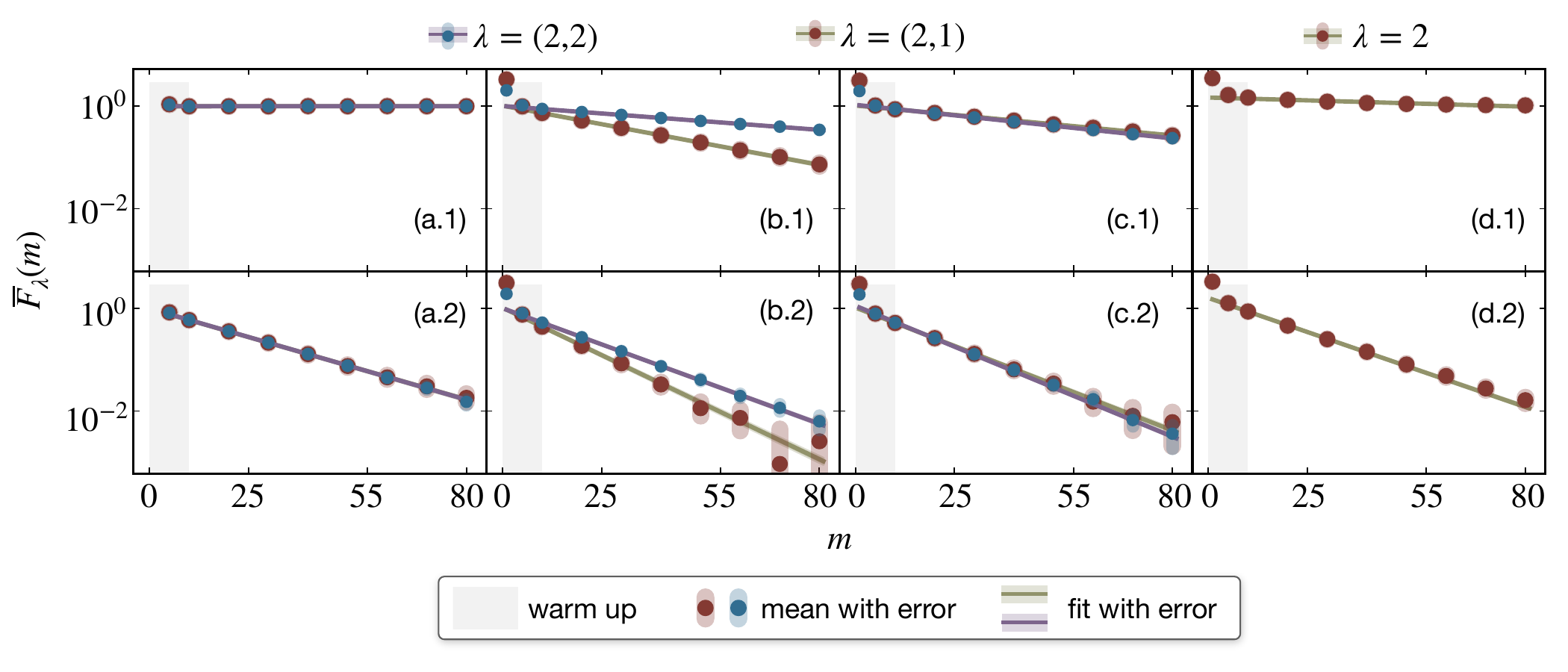}
    \caption{ \emph{Numerical simulations} of the RAB protocol for 4 modes and 2 particles in a 1-dimensional lattice with the uniform setting.
    Top: 4 different noise models, bottom: noise model from above with additional random depolarizing noise with depolarizing strength of $p=0.05$.
    The grey area indicates the warm-up phase.
    (a.1) No noise in a non-interacting system,
    (b.1) time over-evolution of $10\%$ in a non-interacting system,
    (c.1) assumed a non-interacting system with actual interactions of $-0.1$,
    (g) interacting system with less interaction than intended.}
    \label{fig:numerics-8plots}
\end{figure}

\begin{table}[ht!]
\begin{tabular}{r|llll|llll|llll|ll}
\toprule
Err. Model &   \multicolumn{4}{c|}{}   &   \multicolumn{4}{c|}{over evolution} &\multicolumn{4}{c|}{unwanted interaction} & \multicolumn{2}{c}{less interaction} \\
Depol. Err.& $\times$ & $\times$ & $\checkmark$ & $\checkmark$ & $\times$ &  $\times$  & $\checkmark$  & $\checkmark$ & $\times$ & $\times$ & $\checkmark$ & $\checkmark$& $\times$ & $\checkmark$ \\
Figure & (a.1) & & (a.2) & & (b.1) &  & (b.2)  & & (c.1) & & (c.2) & & (d.1) & (d.2) \\
$\lambda$& $(2,2)$ &$(2,1)$ & $(2,2)$ & $(2,1)$& $(2,2)$ & $(2,1)$ & $(2,2)$  & $(2,1)$ & $(2,2)$ & $(2,1)$ & $(2,2)$ & $(2,2)$ & $2$ & $2$ \\
\midrule
$A$ & $0.999$ & $1.001$ & $0.99$ & $1.002$ & $1.03$ & $1.012$ & $1.04$ & $1.02$ & $1.03$ & $1.066$ & $1.05$ & $1.133$ & $1.481$ & $1.61$ \\
$\Delta A$ & $0.005$ & $0.003$ & $0.02$ & $0.008$ & $0.02$ & $0.005$ & $0.03$ & $0.02$ & $0.01$ & $0.006$ & $0.02$ & $0.008$ & $0.006$ & $0.02$ \\
$z$ & $1.00003$ & $1.00000$ & $0.9503$ & $0.9501$ & $0.9673$ & $0.9867$ & $0.918$ & $0.9371$ & $0.9833$ & $0.9813$ & $0.9333$ & $0.9291$ & $0.99489$ & $0.9406$ \\
$\Delta z$ & $0.00008$ & $0.00005$ & $0.0007$ & $0.0003$ & $0.0005$ & $0.0002$ & $0.002$ & $0.0006$ & $0.0003$ & $0.0002$ & $0.0009$ & $0.0005$ & $0.00008$ & $0.0006$ \\
 \bottomrule
\end{tabular}
\caption{\label{tab:8plots-fitting-values} \emph{Fitting values} $F_\lambda =_{\text{fit}} A_\lambda z_\lambda^m$ for Fig.~\ref{fig:numerics-8plots}. The fitting was done in Python with \emph{scipy.curve\_fit}.}
\end{table}

\paragraph*{Noise models.}
The noise channel for the \emph{random depolarizing} noise is represented with $p$ the depolarizing strength as
\begin{align}
    \Lambda_p(\rho) = (1-p)\rho + \frac{p}{d_F}\mathbb I.
    \label{eq:rand-depolarizing-noise}
\end{align}

Another noise model is a \emph{time over-evolution'}. 
Here errors originate from imprecise control over their timing of the evolution, resulting in an unintended `over-evolution', where the desired time for the gate operation is shorter than the actual implemented time,
\begin{align}\label{eq:model-mismatch-noise}
    \ee^{-i\Delta tH} \mapsto \ee^{-i(\Delta t + t')H}\,. 
\end{align}
In addition, we considered errors in the simulation due to unwanted interactions during the dynamics
\begin{align}
    H_{\mathrm{nonint}} \mapsto H_{\mathrm{nonint}} + V^{\mathrm{err}}\sum_{i} a_i^\dagger a_i^\dagger a_i a_i.
\end{align}
Also, the opposite noise effect of an undesired isolation, i.e.,  not as much interaction as intended was simulated.  

\paragraph*{Results for different noise models.}
The eight panels of Fig.~\ref{fig:numerics-8plots} show the resulting simulated RB data and fits for for the four described noise models, both with and without additional depolarizing noise. 
Table~\ref{tab:8plots-fitting-values} is showing the extracted model parameters of the fit. 
In panel (a), no noise has been interleaved, the warm-up phase can be observed, before the signal settles at $1$ not decaying further.
Panel (b) uses depolarizing noise, and the signal decays exponetially at the rate of the depolarizing parameter.
In panel (c), the time over-evolution was simulated over-evolving for a constant $10\%$ of the time. 
We observe that the decay rates for the larger irrep $\lambda=(2,2)$ are now faster than for the smaller irrep $\lambda=(2,1)$. 
The separation remains detectable with additional depolarizing noise in panel (d). 
In (e), a non-interacting system with unintended interaction of $V^{\mathrm{err}}=-0.1$  is simulated. 
Also here the decay rates for the two irreps are moderately different. 
In the setting, of running RAB for an interacting Hamiltonian ensemble with $10\%$ smaller interaction coefficients that target, the deviation is detected as a decay. 
\newpage

\begin{table}[htbp]
\parbox{.6\linewidth}{
\centering
\setlength{\extrarowheight}{3pt}
\begin{tabular}{|c c c c||c c c c||c c c c||c c c c|}
\hline
$a$ & $a'$ & $a''$ & $C^{a''}_{a,a'}$ & $a$ & $a'$ & $a''$ & $C^{a''}_{a,a'}$ & $a$ & $a'$ & $a''$ & $C^{a''}_{a,a'}$ & $a$ & $a'$ & $a''$ & $C^{a''}_{a,a'}$  \\ \hline
$1$ & $1$ & $1$ & $1$ & $4$ & $4$ & $22$ & ${\sqrt{2}}/{2}$ & $6$ & $9$ & $35$ & ${\sqrt{3}}/{6}$ & $9$ & $2$ & $54$ & ${\sqrt{6}}/{3}$ \\ \hline
$1$ & $2$ & $2$ & $1$ & $4$ & $5$ & $23$ & $1$ & $8$ & $8$ & $35$ & $- {\sqrt{6}}/{6}$ & $10$ & $1$ & $54$ & ${\sqrt{6}}/{6}$ \\ \hline
$1$ & $3$ & $3$ & $1$ & $3$ & $6$ & $24$ & $1$ & $9$ & $7$ & $35$ & ${\sqrt{3}}/{6}$ & $9$ & $3$ & $55$ & ${\sqrt{2}}/{2}$ \\ \hline
$1$ & $4$ & $4$ & $1$ & $4$ & $6$ & $25$ & $1$ & $4$ & $10$ & $36$ & ${\sqrt{6}}/{3}$ & $10$ & $2$ & $55$ & ${\sqrt{2}}/{2}$ \\ \hline
$1$ & $5$ & $5$ & $1$ & $2$ & $7$ & $26$ & ${\sqrt{30}}/{6}$ & $7$ & $9$ & $36$ & ${\sqrt{3}}/{6}$ & $10$ & $3$ & $56$ & $1$ \\ \hline
$1$ & $6$ & $6$ & $1$ & $5$ & $4$ & $26$ & ${\sqrt{15}}/{15}$ & $9$ & $8$ & $36$ & $- {\sqrt{3}}/{6}$ & $8$ & $4$ & $57$ & $1$ \\ \hline
$1$ & $7$ & $7$ & ${\sqrt{6}}/{3}$ & $6$ & $2$ & $26$ & $- {\sqrt{30}}/{30}$ & $10$ & $7$ & $36$ & ${\sqrt{6}}/{6}$ & $8$ & $5$ & $58$ & ${\sqrt{3}}/{3}$ \\ \hline
$2$ & $4$ & $7$ & ${\sqrt{3}}/{6}$ & $7$ & $1$ & $26$ & ${\sqrt{15}}/{15}$ & $5$ & $1$ & $37$ & $1$ & $9$ & $4$ & $58$ & ${\sqrt{6}}/{3}$ \\ \hline
$3$ & $2$ & $7$ & $- {\sqrt{3}}/{6}$ & $2$ & $8$ & $27$ & ${\sqrt{30}}/{6}$ & $5$ & $2$ & $38$ & $1$ & $9$ & $5$ & $59$ & ${\sqrt{6}}/{3}$ \\ \hline
$4$ & $1$ & $7$ & ${\sqrt{6}}/{6}$ & $5$ & $5$ & $27$ & ${\sqrt{15}}/{15}$ & $5$ & $3$ & $39$ & $1$ & $10$ & $4$ & $59$ & ${\sqrt{3}}/{3}$ \\ \hline
$1$ & $8$ & $8$ & ${\sqrt{6}}/{3}$ & $6$ & $3$ & $27$ & $- {\sqrt{15}}/{15}$ & $5$ & $4$ & $40$ & ${\sqrt{15}}/{5}$ & $10$ & $5$ & $60$ & $1$ \\ \hline
$2$ & $5$ & $8$ & ${\sqrt{3}}/{6}$ & $7$ & $2$ & $27$ & ${\sqrt{30}}/{30}$ & $6$ & $2$ & $40$ & ${\sqrt{30}}/{15}$ & $8$ & $6$ & $61$ & $1$ \\ \hline
$3$ & $3$ & $8$ & $- {\sqrt{6}}/{6}$ & $2$ & $9$ & $28$ & ${\sqrt{5}}/{3}$ & $7$ & $1$ & $40$ & $- {2 \sqrt{15}}/{15}$ & $9$ & $6$ & $62$ & $1$ \\ \hline
$4$ & $2$ & $8$ & ${\sqrt{3}}/{6}$ & $3$ & $8$ & $28$ & ${\sqrt{5}}/{6}$ & $5$ & $5$ & $41$ & ${\sqrt{15}}/{5}$ & $10$ & $6$ & $63$ & $1$ \\ \hline
$1$ & $9$ & $9$ & ${\sqrt{6}}/{3}$ & $4$ & $7$ & $28$ & $- {\sqrt{5}}/{6}$ & $6$ & $3$ & $41$ & ${2 \sqrt{15}}/{15}$ & $5$ & $7$ & $64$ & $1$ \\ \hline
$2$ & $6$ & $9$ & ${\sqrt{6}}/{6}$ & $5$ & $6$ & $28$ & ${2 \sqrt{5}}/{15}$ & $7$ & $2$ & $41$ & $- {\sqrt{30}}/{15}$ & $5$ & $8$ & $65$ & $1$ \\ \hline
$3$ & $5$ & $9$ & $- {\sqrt{3}}/{6}$ & $6$ & $5$ & $28$ & $- {\sqrt{5}}/{30}$ & $5$ & $6$ & $42$ & ${\sqrt{30}}/{10}$ & $5$ & $9$ & $66$ & ${\sqrt{2}}/{2}$ \\ \hline
$4$ & $4$ & $9$ & ${\sqrt{3}}/{6}$ & $7$ & $4$ & $28$ & ${\sqrt{5}}/{30}$ & $6$ & $5$ & $42$ & ${\sqrt{30}}/{10}$ & $6$ & $8$ & $66$ & ${1}/{2}$ \\ \hline
$1$ & $10$ & $10$ & ${\sqrt{10}}/{5}$ & $8$ & $3$ & $28$ & $- {\sqrt{5}}/{15}$ & $7$ & $4$ & $42$ & $- {\sqrt{30}}/{10}$ & $7$ & $7$ & $66$ & $- {1}/{2}$ \\ \hline
$2$ & $9$ & $10$ & ${2 \sqrt{10}}/{15}$ & $9$ & $2$ & $28$ & ${\sqrt{5}}/{15}$ & $8$ & $3$ & $42$ & ${\sqrt{30}}/{30}$ & $6$ & $7$ & $67$ & $1$ \\ \hline
$3$ & $8$ & $10$ & $- {2 \sqrt{10}}/{15}$ & $10$ & $1$ & $28$ & $- {\sqrt{5}}/{15}$ & $9$ & $2$ & $42$ & $- {\sqrt{30}}/{30}$ & $6$ & $8$ & $68$ & ${\sqrt{2}}/{2}$ \\ \hline
$4$ & $7$ & $10$ & ${2 \sqrt{10}}/{15}$ & $3$ & $7$ & $29$ & ${\sqrt{30}}/{6}$ & $10$ & $1$ & $42$ & ${\sqrt{30}}/{30}$ & $7$ & $7$ & $68$ & ${\sqrt{2}}/{2}$ \\ \hline
$5$ & $6$ & $10$ & ${\sqrt{10}}/{30}$ & $6$ & $4$ & $29$ & ${\sqrt{30}}/{30}$ & $6$ & $1$ & $43$ & $1$ & $7$ & $8$ & $69$ & $1$ \\ \hline
$6$ & $5$ & $10$ & $- {\sqrt{10}}/{30}$ & $8$ & $2$ & $29$ & $- {\sqrt{15}}/{15}$ & $6$ & $2$ & $44$ & ${\sqrt{6}}/{3}$ & $6$ & $9$ & $70$ & ${\sqrt{3}}/{2}$ \\ \hline
$7$ & $4$ & $10$ & ${\sqrt{10}}/{30}$ & $9$ & $1$ & $29$ & ${\sqrt{15}}/{15}$ & $7$ & $1$ & $44$ & ${\sqrt{3}}/{3}$ & $8$ & $8$ & $70$ & ${\sqrt{6}}/{6}$ \\ \hline
$8$ & $3$ & $10$ & ${\sqrt{10}}/{30}$ & $3$ & $8$ & $30$ & ${\sqrt{15}}/{6}$ & $6$ & $3$ & $45$ & ${\sqrt{3}}/{3}$ & $9$ & $7$ & $70$ & $- {\sqrt{3}}/{6}$ \\ \hline
$9$ & $2$ & $10$ & $- {\sqrt{10}}/{30}$ & $4$ & $7$ & $30$ & ${\sqrt{15}}/{6}$ & $7$ & $2$ & $45$ & ${\sqrt{6}}/{3}$ & $7$ & $9$ & $71$ & ${\sqrt{3}}/{2}$ \\ \hline
$10$ & $1$ & $10$ & ${\sqrt{10}}/{30}$ & $6$ & $5$ & $30$ & ${\sqrt{15}}/{30}$ & $7$ & $3$ & $46$ & $1$ & $9$ & $8$ & $71$ & ${\sqrt{3}}/{6}$ \\ \hline
$2$ & $1$ & $11$ & $1$ & $7$ & $4$ & $30$ & ${\sqrt{15}}/{30}$ & $6$ & $4$ & $47$ & ${2 \sqrt{5}}/{5}$ & $10$ & $7$ & $71$ & $- {\sqrt{6}}/{6}$ \\ \hline
$2$ & $2$ & $12$ & $1$ & $8$ & $3$ & $30$ & $- {\sqrt{15}}/{15}$ & $8$ & $2$ & $47$ & ${\sqrt{10}}/{10}$ & $8$ & $7$ & $72$ & $1$ \\ \hline
$2$ & $3$ & $13$ & $1$ & $10$ & $1$ & $30$ & ${\sqrt{15}}/{15}$ & $9$ & $1$ & $47$ & $- {\sqrt{10}}/{10}$ & $8$ & $8$ & $73$ & ${\sqrt{3}}/{3}$ \\ \hline
$2$ & $4$ & $14$ & ${\sqrt{3}}/{2}$ & $4$ & $8$ & $31$ & ${\sqrt{30}}/{6}$ & $6$ & $5$ & $48$ & ${\sqrt{10}}/{5}$ & $9$ & $7$ & $73$ & ${\sqrt{6}}/{3}$ \\ \hline
$3$ & $2$ & $14$ & ${\sqrt{3}}/{6}$ & $7$ & $5$ & $31$ & ${\sqrt{30}}/{30}$ & $7$ & $4$ & $48$ & ${\sqrt{10}}/{5}$ & $9$ & $8$ & $74$ & ${\sqrt{6}}/{3}$ \\ \hline
$4$ & $1$ & $14$ & $- {\sqrt{6}}/{6}$ & $9$ & $3$ & $31$ & $- {\sqrt{15}}/{15}$ & $8$ & $3$ & $48$ & ${\sqrt{10}}/{10}$ & $10$ & $7$ & $74$ & ${\sqrt{3}}/{3}$ \\ \hline
$2$ & $5$ & $15$ & ${\sqrt{3}}/{2}$ & $10$ & $2$ & $31$ & ${\sqrt{15}}/{15}$ & $10$ & $1$ & $48$ & $- {\sqrt{10}}/{10}$ & $10$ & $8$ & $75$ & $1$ \\ \hline
$3$ & $3$ & $15$ & ${\sqrt{6}}/{6}$ & $3$ & $9$ & $32$ & ${\sqrt{30}}/{6}$ & $7$ & $5$ & $49$ & ${2 \sqrt{5}}/{5}$ & $8$ & $9$ & $76$ & $1$ \\ \hline
$4$ & $2$ & $15$ & $- {\sqrt{3}}/{6}$ & $6$ & $6$ & $32$ & ${\sqrt{15}}/{15}$ & $9$ & $3$ & $49$ & ${\sqrt{10}}/{10}$ & $9$ & $9$ & $77$ & $1$ \\ \hline
$2$ & $6$ & $16$ & ${\sqrt{2}}/{2}$ & $8$ & $5$ & $32$ & $- {\sqrt{15}}/{15}$ & $10$ & $2$ & $49$ & $- {\sqrt{10}}/{10}$ & $10$ & $9$ & $78$ & $1$ \\ \hline
$3$ & $5$ & $16$ & ${1}/{2}$ & $9$ & $4$ & $32$ & ${\sqrt{30}}/{30}$ & $6$ & $6$ & $50$ & ${\sqrt{15}}/{5}$ & $5$ & $10$ & $79$ & $1$ \\ \hline
$4$ & $4$ & $16$ & $- {1}/{2}$ & $4$ & $9$ & $33$ & ${\sqrt{30}}/{6}$ & $8$ & $5$ & $50$ & ${2 \sqrt{15}}/{15}$ & $6$ & $10$ & $80$ & $1$ \\ \hline
$3$ & $1$ & $17$ & $1$ & $7$ & $6$ & $33$ & ${\sqrt{15}}/{15}$ & $9$ & $4$ & $50$ & $- {\sqrt{30}}/{15}$ & $7$ & $10$ & $81$ & $1$ \\ \hline
$3$ & $2$ & $18$ & ${\sqrt{6}}/{3}$ & $9$ & $5$ & $33$ & $- {\sqrt{30}}/{30}$ & $7$ & $6$ & $51$ & ${\sqrt{15}}/{5}$ & $8$ & $10$ & $82$ & $1$ \\ \hline
$4$ & $1$ & $18$ & ${\sqrt{3}}/{3}$ & $10$ & $4$ & $33$ & ${\sqrt{15}}/{15}$ & $9$ & $5$ & $51$ & ${\sqrt{30}}/{15}$ & $9$ & $10$ & $83$ & $1$ \\ \hline
$3$ & $3$ & $19$ & ${\sqrt{3}}/{3}$ & $2$ & $10$ & $34$ & ${\sqrt{6}}/{3}$ & $10$ & $4$ & $51$ & $- {2 \sqrt{15}}/{15}$ & $10$ & $10$ & $84$ & $1$ \\ \hline
$4$ & $2$ & $19$ & ${\sqrt{6}}/{3}$ & $5$ & $9$ & $34$ & ${\sqrt{6}}/{6}$ & $8$ & $1$ & $52$ & $1$ &  &  &  &  \\ \hline
$4$ & $3$ & $20$ & $1$ & $6$ & $8$ & $34$ & $- {\sqrt{3}}/{6}$ & $8$ & $2$ & $53$ & ${\sqrt{2}}/{2}$ &  &  &  &  \\ \hline
$3$ & $4$ & $21$ & $1$ & $7$ & $7$ & $34$ & ${\sqrt{3}}/{6}$ & $9$ & $1$ & $53$ & ${\sqrt{2}}/{2}$ &  &  &  &  \\ \hline
$3$ & $5$ & $22$ & ${\sqrt{2}}/{2}$ & $3$ & $10$ & $35$ & ${\sqrt{6}}/{3}$ & $8$ & $3$ & $54$ & ${\sqrt{6}}/{6}$ &  &  &  &  \\ \hline
\end{tabular}
}
\hfill
\parbox{.3\linewidth}{
\centering
\setlength{\extrarowheight}{3pt}
\begin{tabular}{|c c c c||c c c c|}
\hline
$a$ & $a'$ & $a''$ & $C^{a''}_{a,a'}$ & $a$ & $a'$ & $a''$ & $C^{a''}_{a,a'}$  \\ \hline
$1$ & $7$ & $1$ & $- {\sqrt{3}}/{3}$ & $9$ & $2$ & $7$ & ${1}/{3}$ \\ \hline
$2$ & $4$ & $1$ & ${\sqrt{6}}/{6}$ & $10$ & $1$ & $7$ & $- {1}/{3}$ \\ \hline
$3$ & $2$ & $1$ & $- {\sqrt{6}}/{6}$ & $3$ & $7$ & $8$ & $- {\sqrt{6}}/{6}$ \\ \hline
$4$ & $1$ & $1$ & ${\sqrt{3}}/{3}$ & $6$ & $4$ & $8$ & ${\sqrt{6}}/{6}$ \\ \hline
$1$ & $8$ & $2$ & $- {\sqrt{3}}/{3}$ & $8$ & $2$ & $8$ & $- {\sqrt{3}}/{3}$ \\ \hline
$2$ & $5$ & $2$ & ${\sqrt{6}}/{6}$ & $9$ & $1$ & $8$ & ${\sqrt{3}}/{3}$ \\ \hline
$3$ & $3$ & $2$ & $- {\sqrt{3}}/{3}$ & $3$ & $8$ & $9$ & $- {\sqrt{3}}/{6}$ \\ \hline
$4$ & $2$ & $2$ & ${\sqrt{6}}/{6}$ & $4$ & $7$ & $9$ & $- {\sqrt{3}}/{6}$ \\ \hline
$1$ & $9$ & $3$ & $- {\sqrt{3}}/{3}$ & $6$ & $5$ & $9$ & ${\sqrt{3}}/{6}$ \\ \hline
$2$ & $6$ & $3$ & ${\sqrt{3}}/{3}$ & $7$ & $4$ & $9$ & ${\sqrt{3}}/{6}$ \\ \hline
$3$ & $5$ & $3$ & $- {\sqrt{6}}/{6}$ & $8$ & $3$ & $9$ & $- {\sqrt{3}}/{3}$ \\ \hline
$4$ & $4$ & $3$ & ${\sqrt{6}}/{6}$ & $10$ & $1$ & $9$ & ${\sqrt{3}}/{3}$ \\ \hline
$1$ & $10$ & $4$ & $- {\sqrt{2}}/{2}$ & $4$ & $8$ & $10$ & $- {\sqrt{6}}/{6}$ \\ \hline
$2$ & $9$ & $4$ & ${\sqrt{2}}/{6}$ & $7$ & $5$ & $10$ & ${\sqrt{6}}/{6}$ \\ \hline
$3$ & $8$ & $4$ & $- {\sqrt{2}}/{6}$ & $9$ & $3$ & $10$ & $- {\sqrt{3}}/{3}$ \\ \hline
$4$ & $7$ & $4$ & ${\sqrt{2}}/{6}$ & $10$ & $2$ & $10$ & ${\sqrt{3}}/{3}$ \\ \hline
$5$ & $6$ & $4$ & ${\sqrt{2}}/{6}$ & $3$ & $9$ & $11$ & $- {\sqrt{6}}/{6}$ \\ \hline
$6$ & $5$ & $4$ & $- {\sqrt{2}}/{6}$ & $6$ & $6$ & $11$ & ${\sqrt{3}}/{3}$ \\ \hline
$7$ & $4$ & $4$ & ${\sqrt{2}}/{6}$ & $8$ & $5$ & $11$ & $- {\sqrt{3}}/{3}$ \\ \hline
$8$ & $3$ & $4$ & ${\sqrt{2}}/{6}$ & $9$ & $4$ & $11$ & ${\sqrt{6}}/{6}$ \\ \hline
$9$ & $2$ & $4$ & $- {\sqrt{2}}/{6}$ & $4$ & $9$ & $12$ & $- {\sqrt{6}}/{6}$ \\ \hline
$10$ & $1$ & $4$ & ${\sqrt{2}}/{6}$ & $7$ & $6$ & $12$ & ${\sqrt{3}}/{3}$ \\ \hline
$2$ & $7$ & $5$ & $- {\sqrt{6}}/{6}$ & $9$ & $5$ & $12$ & $- {\sqrt{6}}/{6}$ \\ \hline
$5$ & $4$ & $5$ & ${\sqrt{3}}/{3}$ & $10$ & $4$ & $12$ & ${\sqrt{3}}/{3}$ \\ \hline
$6$ & $2$ & $5$ & $- {\sqrt{6}}/{6}$ & $2$ & $10$ & $13$ & $- {\sqrt{3}}/{3}$ \\ \hline
$7$ & $1$ & $5$ & ${\sqrt{3}}/{3}$ & $5$ & $9$ & $13$ & ${\sqrt{3}}/{3}$ \\ \hline
$2$ & $8$ & $6$ & $- {\sqrt{6}}/{6}$ & $6$ & $8$ & $13$ & $- {\sqrt{6}}/{6}$ \\ \hline
$5$ & $5$ & $6$ & ${\sqrt{3}}/{3}$ & $7$ & $7$ & $13$ & ${\sqrt{6}}/{6}$ \\ \hline
$6$ & $3$ & $6$ & $- {\sqrt{3}}/{3}$ & $3$ & $10$ & $14$ & $- {\sqrt{3}}/{3}$ \\ \hline
$7$ & $2$ & $6$ & ${\sqrt{6}}/{6}$ & $6$ & $9$ & $14$ & ${\sqrt{6}}/{6}$ \\ \hline
$2$ & $9$ & $7$ & $- {1}/{3}$ & $8$ & $8$ & $14$ & $- {\sqrt{3}}/{3}$ \\ \hline
$3$ & $8$ & $7$ & $- {1}/{6}$ & $9$ & $7$ & $14$ & ${\sqrt{6}}/{6}$ \\ \hline
$4$ & $7$ & $7$ & ${1}/{6}$ & $4$ & $10$ & $15$ & $- {\sqrt{3}}/{3}$ \\ \hline
$5$ & $6$ & $7$ & ${2}/{3}$ & $7$ & $9$ & $15$ & ${\sqrt{6}}/{6}$ \\ \hline
$6$ & $5$ & $7$ & $- {1}/{6}$ & $9$ & $8$ & $15$ & $- {\sqrt{6}}/{6}$ \\ \hline
$7$ & $4$ & $7$ & ${1}/{6}$ & $10$ & $7$ & $15$ & ${\sqrt{3}}/{3}$ \\ \hline
$8$ & $3$ & $7$ & $- {1}/{3}$ &  &  &  &  \\ \hline
\end{tabular}
}
\label{tab:cgc_4d_2d}
\caption{Clebsch-Gordan-Coefficients for systems with dimension $d=4$ and particle number $n=2$ for irrep $\lambda = (2, 2)$ on the LHS and  $\lambda = (2, 1)$ for the RHS.}
\end{table}

\ifjournal
\end{bibunit}
\fi
\end{document}